\documentclass[smallextended]{svjour3}       % onecolumn (second format)
\smartqed  % flush right qed marks, e.g. at end of proof
\usepackage{graphicx}
\usepackage{epsfig, amsmath, amsfonts, amssymb, multirow,algorithm,algorithmic}
\usepackage{graphicx}
\usepackage{setspace}
\usepackage{enumerate}
\usepackage{lipsum}
\usepackage{epstopdf}
\usepackage{bm}
\usepackage{natbib}
\usepackage{subfig}
\usepackage{hyperref}
\hypersetup{citecolor=blue,colorlinks=true}
\newcommand{\beq}{\begin{equation}}
\newcommand{\eeq}{\end{equation}}
\newcommand{\ber}{\begin{eqnarray}}
\newcommand{\eer}{\end{eqnarray}}

\newcommand{\x}{\boldsymbol{x}}
\newcommand{\y}{\boldsymbol{y}}
\newcommand{\z}{\boldsymbol{z}}
\newcommand{\bs}{\boldsymbol}

\newcommand{\bit}{\begin{itemize}}
\newcommand{\eit}{\end{itemize}}
\newcommand{\iid}{\stackrel{\mathrm{i.i.d.}}{\sim}}
\newcommand{\ben}{\begin{enumerate}}
\newcommand{\een}{\end{enumerate}}

\numberwithin{equation}{section}

\begin{document}

%\begin{titlepage}
%
%\setlength{\topmargin}{4cm}
%\begin{center}
%\textbf{\huge New Approaches of Latin Hypercube Sampling for Dependent Random Variables %}
%\end{center}
%\end{titlepage}
%\maketitle
%

% -------------------------------------------------------
%\begin{frontmatter}

\title{A Two Stage Adaptive Metropolis Algorithm }

%\author{Anirban Mondal, Kai Yin}
%\address{Case Western Reserve University, Cleveland, OH 44106, USA}
%\author{Abhijit Mandal}
%\address{University of Texas at El Paso, El Paso, TX 79968, USA}

\author{Anirban Mondal         \and
        Kai Yin  \and
        Abhijit Mandal
}


\institute{A. Mondal \at
              Case Western Reserve University, Cleveland, OH 44106, USA 
              \email{axm912@case.edu}           %  \\
           \and
           K. Yin \at
             Case Western Reserve University, Cleveland, OH 44106, USA 
             \email{kxy160@case.edu}   
             \and
             A. Mandal \at
             University of Texas at El Paso, El Paso, TX 79968, USA
             \email{amandal@utep.edu}
}

\maketitle 

\begin{abstract} 
% One of the main factor for the convergence of Markov chain Monte Carlo methods such as Metropolis Hastings algorithm is the choice of a proposal distribution. In Adaptive Metropolis (AM) algorithm the proposal distribution is updated along the process using the full information cumulated so far. Due to
% the adaptive nature of the process, the AM algorithm is non-Markovian, but we establish here that it
% has the correct ergodic properties. We also include the results of our numerical tests, which indicate
% that the AM algorithm competes well with traditional Metropolis±Hastings algorithms, and
% demonstrate that the AM algorithm is easy to use in practical computation.

% Due to the adaptive nature of the proposals, the adaptive Metropolis algorithm have good convergence properties, specially for high-dimensional setting. But it is not computationally efficient for posterior sampling in Bayesian inference when the likelihood is expensive. The two stage Metropolis Hastings algorithm screens the bad proposal in the first stage using an approximate inexpensive likelihood, but it suffers from convergence issues. 

We propose a new sampling algorithm combining two quite powerful ideas in the Markov chain Monte Carlo literature -- adaptive Metropolis sampler and two-stage Metropolis-Hastings sampler. 
The proposed sampling method will be particularly very useful for high-dimensional posterior sampling in Bayesian models with expensive likelihoods. In the first stage of the proposed algorithm, an adaptive proposal is used based on the previously sampled states and the corresponding acceptance probability is computed based on an approximated inexpensive target density. The true expensive target density is evaluated while computing the second stage acceptance probability only if the proposal is accepted in the first stage. The adaptive nature of the algorithm guarantees faster convergence of the chain and very good mixing properties. On the other hand, the  two-stage  approach helps in rejecting the bad proposals in the inexpensive first stage, making the algorithm computationally efficient. As the proposals are dependent on the previous states the chain loses its Markov property, but we prove that it retains the desired ergodicity property. The performance of the proposed algorithm is compared with the existing algorithms in two simulated and two real data examples.

\end{abstract}

%\begin{keyword} 
\keywords{Markov chain Monte Carlo, Metropolis-Hastings, Two-stage Metropolis-Hastings, Adaptive Metropolis, Ergodicity. }
%\end{keyword}

%\end{frontmatter}

%\linenumbers
% -----------------------------------------------------

\section{Introduction}
%\section{Metropolis Hastings Algorithm}
%\section{Two Stage MCMC Algorithm }
%\section{Adaptive MCMC Algorithm}
Markov Chain Monte Carlo (MCMC) methods have been widely used to sample from a target probability distribution for which the density function can only be analytically expressed up to a constant of proportionality, or  for which it is difficult to compute the inverse of the corresponding cumulative distribution function. Due to the recent advances in Bayesian methodology, posterior sampling has become an integral part of Bayesian inferential procedures. The complicated hierarchical nature of advanced Bayesian models more often result in intractable posteriors. The use of MCMC methods is thus essential for sampling-based Bayesian inference related to such models.
Metropolis algorithm \citep{metro1953} and the Metropolis-Hastings (MH) algorithm \citep{hastings1970} are the two most popular MCMC methods used to sample from an unknown target probability distribution. The main difficulty of the random walk Metropolis algorithm is to choose an effective proposal distribution such that reasonable results are obtained by simulation in a limited amount of time. Since the target density is unknown, the size and the spatial orientation of an effective proposal distribution are often very difficult to choose, especially for high-dimensional target probability distributions  \citep{gilks1995adaptive, Gelman1996, Roberts1997, haario1999adaptive}. Several research efforts have been made in this regard, mainly based on adaptive MCMC methods that use the history of the process to effectively `tune' the proposal distribution \citep{Evans1991, Gilks1994, gelfand1994markov, Gilks1998}.
%(see Gilks et al. 1994, 1998, Evans (1991), Fishman (1996), Gelfand and Sahu (1994), Gilks and Roberts (1995) ).
\cite{haario2001} introduced an adaptive Metropolis (AM) algorithm which adapts continuously to the target distribution by using a proposal distribution with covariance estimated from the previous sampled states. The adaptation significantly affects both the size and the spatial orientation of the proposal distribution.
The rapid start of the adaptation in the AM method ensures that the search becomes more effective at an early stage of the simulation. As a result, the chain converges faster and it also has very good mixing properties when compared to other MCMC algorithms, especially for high dimensional target distribution.

% In the AM algorithm the covariance of the proposal
% distribution is calculated using all of the previous states.
% The method is easily implemented with no extra computational cost since one may apply a simple recursion formula for the covariances involved.
% An important advantage of the AM algorithm is that it starts using the cumulating information right at the beginning of the simulation. The rapid start of the adaptation ensures that the search becomes more effective at an early stage of the simulation, which diminishes the number of function evaluations needed. Due to the adaptive nature of the proposal the AM converges faster and it also have very good mixing properties when compared to other MCMC algorithms, especially for high dimensional target distribution.
% become the most effective and popular algorithm in sampling from a high dimensional target distribution.

% In the past twenty-five years, Bayesian statistics have become increasingly popular as they are capable of analyzing data with complex structures. Consequently, Bayesian methods have been proven to be effective in a wide range of applications. The
% rise in popularity is largely attributed to simulation based algorithms which can approximate the complex posterior distributions of non-conjugate models, such as Markov Chain Monte Carlo (MCMC) methods including the Metropolis-Hastings
% (MH) algorithm (Robert and Casella 2013).
% The term “tall data” generally describes data in which n >> p, that is, when the
% number of observations is much larger than the number of predictors. 
In the last few decades, Bayesian statistics have become increasingly popular as they are capable of analyzing data with complex structures. For posterior sampling-based inference in Bayesian hierarchical models, the AM has emerged as one the popular and effective sampling algorithms, especially when the dimension of the parameter is very high.
However, in such posterior sampling, the direct use of the MH or the AM algorithm often becomes computationally challenging due to repeated evaluation of expensive likelihood functions. Such models with expensive likelihoods can be found in many situations -- a Bayesian hierarchical model with a very large observed data (also called tall data), a Bayesian inverse problem where likelihood contains a non-linear ``forward model" etc. In the first situation, a very large number of observations makes the likelihood  computationally expensive  as a complete scan of the data is needed to compute the joint density of the data given the parameters. In the latter case, the ``forward model" is generally described by a system of ordinary differential equation (ODE) or partial differential equation (PDE) which is solved by a discretized numerical method. This numerical solver can be very expensive depending on the complexity of the model and the discretization grid. The computation complexity of the forward model makes the likelihood computationally expensive. 
In such complex statistical models, the repeated evaluations of the likelihood in each iteration of the sampling algorithm make the MCMC method computationally very expensive and often infeasible. One of the commonly used posterior sampling techniques that deal with such expensive likelihood evaluations is the Approximate Bayesian Computation (ABC) based methods \citep{Tavare1997, Beaumont2002, Joyce2008}. ABC-based methods approximate the likelihood function by simulations, the outcomes of which are compared with the observed data. Sufficient summary statistics are also used instead of the full high dimensional data for the comparison. The second approach  is to use sub-sampling methods to provide a faster estimation of the likelihood \citep{Bard2017, Mat2019}. Another very popular approach is the two-stage Metropolis-Hastings (TSMH) method \citep{Christin2005, efendiev2005, mondal2014, payne2018},  which uses an approximate cheap likelihood to screen the bad proposals on the first stage. The expensive true likelihood is only evaluated for a fraction of total iterations when the good proposals pass the first stage. However, the main difficulty in the TSMH method, in a high-dimensional setting, is in choosing an appropriate effective proposal distribution that can ensure good mixing properties and faster convergence of the chain.
%\cite{}, \cite{mondal2014}, 
%\cite{}. 

% For MCMC methods related to posterior sampling , as the sample size increases, so does the computational demand of the algorithm. Specifically, for MH, the increased computational demand is driven by the complete scan of the data through likelihood evaluations on each iteration of the algorithm. Moreover 

For high-dimensional expensive target densities, such as posteriors with expensive likelihoods, we propose an algorithm combining the two powerful ideas in the MCMC literature that are discussed above: i) adaptive Metropolis (AM) method and ii) two-stage Metropolis-Hasting (TSMH) method. The proposed algorithm is named the two-stage adaptive Metropolis (TSAM) algorithm. At each iteration of the algorithm, in the first stage, an adaptive proposal is drawn from a distribution whose covariance is estimated from all the previous sampled states. The acceptance probability of this proposal in the first stage is computed based on an approximate version of the target density that is computationally inexpensive. If the proposal is accepted in this first stage, then it is passed on to the second stage, where another acceptance probability is computed based on the true expensive target density. The adaptation strategy in the first stage forces the proposal distribution to approach an appropriately scaled Gaussian approximation of the target distribution, which ensures faster convergence of the chain and very good mixing properties, especially in a high dimensional setting. On the other hand, the two-stage nature of the algorithm makes it computationally efficient, when compared to the AM, as a result of rejecting the bad proposals in the inexpensive first stage. The TSAM will be particularly very useful for high-dimensional posterior sampling in Bayesian models with expensive likelihoods. As the proposal distribution in TSAM depends on all the previous sampled states the chain loses its Markov property. But our main result, Theorem \ref{thm1}, proves that the TSAM  produces a chain that has the correct ergodicity properties, assuming that the target density is bounded from above and has bounded support.

% The adaptive nature of the algorithm in the first stage guarantees faster convergence of the chain and very good mixing properties, specially in a high dimensional setting. On the other hand, the two-stage nature of the algorithm makes it computationally efficient, when compared to the AM, as a result of rejecting the bad proposals in the inexpensive first stage. As the proposal distribution in TSAM depends on all the previous sampled states the chain losses its Markov property. But our main result, Theorem \ref{thm1}, proves that the TSAM  produces a chain which have the correct ergodic properties, assuming that the target density is bounded from above and has a bounded support. 

The proposed algorithm is first applied to a few simulation examples, viz., sampling from a multivariate t-distribution and sampling from a transformed multivariate normal distribution with a banana-shaped relationship between components. The simulation results show that both the TSAM and AM algorithms perform comparably in terms of effective sampling and ergodicity and both are much superior in this regard to the regular MH  and TSMH algorithms. To compare the computational efficiency of the TSAM and AM methods, both are used for posterior sampling in a Bayesian logistic regression model with a very large banking dataset. As the dataset is large the full likelihood needs a complete scan of the data and hence becomes very expensive. For such a logistic regression model with a large number of zero responses, an approximate inexpensive likelihood can be computed using a sub-sample method \cite{raftery2012fast}. This approximate likelihood is used in the first stage of TSAM. The efficiency of the two methods is compared using a statistic named effective draws per minute (EPDM), which incorporates both execution time and autocorrelation of the chain. 
Finally, both the TSAM and AM methods are used for posterior sampling in Bayesian calibration of a predator-prey model. This is an inverse problem where the ``forward model" is a system of ODE which is expensive to solve numerically in a fine time grid. The Bayesian calibration method casts the inverse solution in terms of a posterior probability distribution. The likelihood term includes the nonlinear forward model and hence the posterior is intractable. At the first stage of TSAM, an approximate likelihood based on the solution to the forward model on a coarser time grid is used. If the proposal is passed in the first stage, the forward model is solved in the fine grid to compute the more accurate likelihood. The simulation results from these last two examples showed that the TSAM is computationally more efficient than the AM, the level of efficiency depending on relative efficiency in computing the approximate likelihood with respect to the actual likelihood, while each having a similar convergence and ergodicity properties. 

The article is organized as follows, in Section \ref{TSAM} we introduce the proposed TSAM algorithm and in Section \ref{Ergodicity} we prove that the proposed algorithm produces a chain that has the correct ergodicity properties. Application of the algorithm to simulation examples and posterior sampling in Bayesian methods using real data are being discussed in Section \ref{simulation}. In Section \ref{conclusion}, we conclude our article with a brief discussion.

%\newpage
\section{Description of the two-stage adaptive Metropolis algorithm}
\label{TSAM}
Suppose our goal is to sample from a $d$-dimensional target distribution $\pi(\x)$, which is expensive to evaluate even to a constant of proportionality. Let us assume $\pi(\x)$ can be well approximated by $\pi^*(\x)$ which is much cheaper to compute in terms of computation time.
Suppose we are at iteration $t-1$ and $(\x_0, \x_1, ..., \x_{t-1})$ are the sampled states, $\x_0$ being the initial state. The proposed two stage adaptive Metropolis (TSAM) algorithm is given by the following steps: 
%\begin{algorithm}
%\caption{two-stage Adaptive Metropolis Algorithm}
%\label{tsam}
\begin{enumerate}
\item A candidate $\x^*$ is sampled from the proposal distribution  $q_t(.|\x_0, \x_1, ..., \x_{t-1})$, which is a Gaussian distribution with mean $\x_{t-1}$ and covariance $C_t$ defined as
\beq
\label{ct}
C_t = 
\begin{cases}
C_0 & \text{if} \ t <t_0, \\
s_d {\rm cov}(\x_0, \x_1, ..., \x_{t-1}) + s_d \epsilon \bs{I}_d & \text{if} \ t\geq t_0,
\end{cases}
\eeq
where $s_d$ is a parameter
that depends only on dimension $d$ and $\epsilon$ is a small constant which ensures that $C_t$ does
not become singular. Here $\bs{I}_d$ denotes the $d$-dimensional identity matrix. In order to
start, an arbitrary, strictly positive definite initial covariance $C_0$ is selected according to the
best prior knowledge.  $t_0 \ (>0)$ is the length of
an initial period for which $C_0$ is used for the proposal distribution.

\item  In the first stage the candidate $\x^*$ is screened and passed to the second stage with acceptance probability given by
\beq
\label{stage1ap}
\alpha_1(\x_{t-1}, \x^*) = \min\Big(1, \frac{\pi^*(\x^*)}{\pi^*(\x_{t-1})}\Big).
\eeq
This is equivalent of taking the final proposal as 
\beq
\x = 
\begin{cases}
\x^* & \text{with probability} \ \alpha_1(\x_{t-1}, \x^*), \\
\x_{t-1} & 
\text{with probability} \ 1- \alpha_1(\x_{t-1}, \x^*).
\end{cases}
\eeq
% Therefore, the final proposal $\x$ is generated from the effective instrumental distribution
% \beq
% Q_1(\x|\x_{t-1}) = \alpha_1(\x_{t-1}, \x)q_t(\x|\x_0, \x_1, ..., \x_{t-1}) + \Big(1- \int\alpha_1(\x_{t-1}, \x)q_t(\x|\x_0, \x_1, ..., \x_{t-1})d\x \Big)\delta_{\x_{t-1}}(\x),
% \eeq
% where $\delta_{\x_{t-1}}(\x)$ is the Kronecker delta function.
\item Accept $\x$ as a sample from $\pi$ with acceptance probability 
% \beq
% \alpha_2(\x_{t-1}, \x) = \min \Big(1, \frac{Q_1(\x_{t-1}|\x)\pi(\x)}{Q_1(\x|\x_{t-1})\pi(\x_{t-1})}\Big).
% \eeq
% It can be shown that the final acceptance probability can be simplified to 
\beq
\label{stage2ap}
\alpha_2(\x_{t-1}, \x) = \min \Big(1, \frac{\pi(\x)\pi^*(\x_{t-1})}{\pi(\x_{t-1})\pi^*(\x)} \Big).
\eeq
% The proof is given  in the appendix.

% % \begin{proof}
% % When $x=x_{t-1}$, the proof is trivial as both $\min \Big(1, \frac{Q_1(\x_{t-1}|\x)\pi(\x)}{Q_1(\x|\x_{t-1})\pi(\x_{t-1})}\Big)$ and  $\min \Big(1, \frac{\pi(\x)\pi^*(\x_{t-1})}{\pi(\x_{t-1})\pi^(\x)} \Big)$ are equal to $1$.

% When $x\neq x_{t-1}$, 
% \ber
% Q_1(x_{t-1}|x) &=& \alpha_1(\x, \x_{t-1})q_t(\x_{t-1}|\x)\\
% &=& q_t(\x_{t-1}|\x)\min\Big(1, \frac{\pi^*(\x_{t-1})}{\pi^*(\x)}\Big)\\
% &=&\frac{1}{\pi^*(\x)}\min\Big(q_t(\x_{t-1}|\x)\pi^*(\x),q_t(x_{t-1}|\x) \pi^*(\x_{t-1})\Big)
% &=& \frac{pi^*(\x_{t-1})}{\pi^*(\x)}\min\Big(q_t(\x_{t-1}|\x)\pi^*(\x),q_t(x_{t-1}|\x) \pi^*(\x_{t-1})\Big)
% \eer

% \end{proof}

\end{enumerate}
%\end{algorithm}

%\clearpage
%\begin{remark}
\noindent
{\bf Remark 1:}
% Let us denote the covariance function $cov(\x_0, \x_1, \cdots \x_{t-1})$ by $\tilde{C}_t$. 
The covariance matrix $C_t$ needs to be computed in the first stage at every iteration of the algorithm. But this can be computed cheaply using the recursive formula
% \beq
% C_{t+1} = \frac{t-1}{t}C_t + \frac{s_d}{t} (t\bar{\x}_{t-1}\bar{\x}_{t-1}^T - (t+1)\bar{\x}_{t}\bar{\x}_{t}^T + \x_t \x_t^T +\epsilon I_d),
% \eeq
%
\beq
C_{t+1} = \frac{t-1}{t}C_t + \frac{s_d}{t} (t\bar{\x}_{t-1}\bar{\x}_{t-1}^T - (t+1)\bar{\x}_{t}\bar{\x}_{t}^T + \x_t \x_t^T +\epsilon \bs{I}_d),
\eeq
%
% {\color{red} Verify this formula.. although Harraio used it in his two papers but carefully looking into it it seems some error}
where $\bar{\x}_{t} = \frac{1}{t+1}\sum_{i=0}^{t}\x_i$,  and it can also be updated recursively by $\bar{\x}_t = \frac{1}{t+1}\big(t\bar{\x}_{t-1} + \x_{t}\big)$. Using this, the recursive formula can be simplified as
% {\color{red} isn't it $\bar{\x}_{t} = \frac{1}{t+1}\sum_{i=0}^{t}\x_i$ ?}
%{\color{red} Or equivalently:  
% \beq
% C_{t+1} = \frac{t-1}{t}C_t + \frac{s_d}{t} (t\bar{\x}_{t-1}\bar{\x}_{t-1}^T - \bar{\x}_{t}\bar{\x}_{t}^T + \x_t \x_t^T +\epsilon I_d),
% \eeq
% \beq
% C_{t+1} = \frac{t-1}{t}C_t + \frac{s_d}{t}\big((\x_t-\bar{\x}_{t})^T(\x_t-\bar{\x}_{t}) + \epsilon I_d \big) ,
% \eeq 
% }
% \beq
% C_{t+1} = \frac{t-1}{t}C_t + \frac{s_d(t+1)}{t^2}\big((\x_t-\bar{\x}_{t})(\x_t-\bar{\x}_{t})^T \big) + \frac{s_d}{t} \epsilon I_d,
% \eeq 
% or
\beq
C_{t+1} = \frac{t-1}{t}C_t + \frac{s_d}{t+1}(\x_{t}-\bar{\x}_{t-1})(\x_{t}-\bar{\x}_{t-1})^T  + \frac{s_d}{t} \epsilon \bs{I}_d.
\eeq 
%}
%\end{remark}

\noindent
{\bf Remark 2:}
Computation of $C_{t}^{1/2}$ is required to sample from the normal proposal in the first stage of every iteration. This is usually done by the Cholesky decomposition, which has a computation cost of $O(d^3)$. One way to reduce this computation cost is to apply the rank one Cholesky update of $C_{t}$ at every iteration, as given in  \cite{dongarra1979}, which requires $O(d^2)$ operation. Let $R_t=C_{t}^{1/2}$ as obtained by the Cholesky decomposition, then one can apply the rank one Cholesky update algorithm to 
$\sqrt{\frac{t-1}{t}} R_t$  and  $\sqrt{\frac{s_d}{t+1}}(\x_t-\bar{\x}_{t-1})^T$
and obtain the update $R_{t+1}$ of $C_{t+1}^{1/2}$.

\noindent
{\bf Remark 3:}
 This covariance update can also be done after every $K$-th step using a similar recursive formula
% \beq
%     C_{t+1} = \frac{t-K-1}{K +t-K-1} C_{t-K} + \frac{K (t-K)}{(K+(t-K)-1)(K+(t-K)))}(\x-\bar{\x}_{t-K})^T(\x-\bar{\x}_{t-K}),
% \eeq
%
%
% \beq
%     C_{t+1} = \frac{t-K-1}{t} C_{t-K} + \frac{s_d K (t-K)}{(K+t_{t-K}-1)(K+t_{t-K})}(\x-\bar{\x}_{t-K})^T(\x-\bar{\x}_{t-K}),
% \eeq
\beq
    C_{t+1} = \frac{t-K-1}{t} C_{t-K} + \frac{1}{t}\Big(\sum_{i=t-K+1}^{t}\x_i\x_i^T +(t-K+1)\bar{\x}_{t-K}\bar{\x}_{t-K}^T -(t+1)\bar{\x}_{t}\bar{\x}_{t}^T\Big),
\eeq
and 
% \beq
% \bar{\x}_{t} = \bar{\x}_{t-K} + \frac{K}{K+t_{t-K}}(\x - \bar{\x}_{t-K})
% \eeq
\beq
\bar{\x}_{t} = \frac{1}{t+1}\Big((t-K)\bar{\x}_{t-K} + \sum_{i=t-K+1}^{t}\x_{i}\Big)
\eeq
%
% {\color{red} Verify the previous two equations.. there is serious error here, what is $t_{t-K}$ and $\x$? Also $\x_{t+1}$ is not used here, note that $\x_t = (\x_0, \cdots, x_t)$ with dimension $t+1$}
%
% {\color{blue} right, in Harrio's notation, he used $x_t$) }
%
This will save more computational time  while generating the proposals in the first stage. The corresponding Cholesky update can be obtained by $K$ successive application of rank one Cholesky update.
However, it would be computationally cheaper to calculate the Cholesky decomposition directly on $C_{t+1}$ when the updating step size $K$ is larger than the dimension of parameters $d$.  

\noindent 
{ \bf Remark 4:}
One of the common choices for the scaling parameter is given by $s_d=(2.4)^2/d$. \cite{Gelman1996} showed that in a certain sense this choice optimizes the mixing properties of the Metropolis search in the case of Gaussian targets and Gaussian proposals. In our simulation study, it is observed that choosing an even smaller $s_d$ initially gives better performance. 

\noindent
{\bf Remark 5:}
 $C_0$ should be a positive definite matrix. The choice of $C_0$ should be such that the chain does not get stuck at a given state within the initial $t_0$ runs. If a priori knowledge about the covariance of the target distribution is available, it can be utilized in choosing $C_0$. If no information about the covariance is available, we recommend using $C_0=\bs{I}_dc(2.4)^2/d$, where $c\leq1$ is chosen such that the chain moves more frequently during the initial $t_0$ runs.

\noindent 
{ \bf Remark 6:} The choice for the length of the initial segment $t_0$ reflects our trust in the initial covariance $C_0$. The bigger it is chosen, the slower the effect of the adaptation is felt.

\noindent 
{ \bf Remark 7:}
The role of the parameter $\epsilon$ is to ensure that $C_t$ is a positive definite matrix. In practice, this should be chosen as a very small positive quantity close to $0$. The condition $\epsilon>0$ is also required to prove the correct ergodicity property of the algorithm (see Theorem \ref{thm1}).
%\newpage

\section{Ergodicity of the two-stage adaptive Metropolis algorithm} \label{Ergodicity}
As the proposal distribution depends on all the previous samples via $C_t$, the chain loses the Markov property. In this section, we will prove, from the first principles, that the TSAM algorithm has the right ergodicity properties and hence provides a correct simulation of the target distribution. 
We shall assume that $D \subset \mathbb{R}^d$ is a Borel-measurable subset of the Euclidean space, and the target $\pi : D \rightarrow [0, \infty)$ is a probability density on D.
Thus $D$ is our state-space equipped with the Borel $\sigma$-algebra $\mathcal{B}(D)$ and we denote $\mathcal{M}(D)$ as the set of finite measures on $(D, \mathcal{B}(D))$. 
We also assume that the density is bounded from above on $D$: for some $M<\infty$, i.e., we have that
\beq
\label{boundpi}
\pi(\x) < M \ \text{for all} \ \x \in D.
\eeq
Let $q_C(\x,\y)$ be the density of a Gaussian proposal (proposing $\y$ from the current state $\x$) with covariance matrix $C$. Thus 
\beq
q_C(\x,\y)=\frac{1}{\sqrt{2\pi}|C|^{1/2}} \exp\Big(-\frac{1}{2}(\y-\x)^TC^{-1}(\y-\x)\Big).
\eeq
Then the two-stage adaptive Metropolis algorithm transition probability, $Q_C$, having the target density $\pi(\x)$, the approximate first stage density $\pi^*(\x)$, and the proposal density $q_C$ is given by:
For any Borel-measurable subset $A \subset D$ such that $\x \notin A$,
\beq
Q_C(\x;A)=\int_{A}q_C(\x,\y)\alpha_1(\x,\y)\alpha_2(\x,\y)d\y
\eeq
and
\beq
Q_C(\x;\{\x \}) = 1- Q_C(\x:D \backslash \{\x\})
%=\int_{R^d}q_C(x,y)(1-\alpha_1(x,y))dy + \int_{R^d}q_C(x,y)\alpha_1(x,y)(1-\alpha_2(x,y)dy%
\eeq
or equivalently for any Borel measurable set $A\subset D$
\ber
\label{transk}
Q_C(\x;A)&=&\int_{A}q_C(\x,\y)\alpha_1(\x,\y)\alpha_2(\x,\y)d\y \nonumber\\
&+&\chi_A(\x)\left(\int_{\mathbb{R}^d}q_C(\x,\y)(1-\alpha_1(\x,\y))d\y +  \int_{\mathbb{R}^d}q_C(\x,\y)\alpha_1(\x,\y)(1-\alpha_2(\x,\y))d\y\right),
\eer
where $\chi_A(\x)$ is the characteristic function of the set $A$, $\alpha_1(\x,\y)=\min\Big(1,\frac{\pi^*(\y)}{\pi^*(\x)}\Big)$ and $\alpha_2(\x,\y)=\min \Big(1,\frac{\pi^(\y)\pi^*(\x)}{\pi^(\x)\pi^*(\y)}\Big)$.
The sequence $(K_n)_{n=1}^{\infty}$ of generalized transition probabilities defining the two-stage adaptive Metropolis algorithm is given by
\beq
\label{tp}
K_n(\x_0,\x_1,\cdots, \x_{n-1};A)=Q_{C_n(\x_0,\x_1,\cdots,\x_{n-1})}(\x_{n-1};A)
\eeq

It is also useful to define the transition probability that is obtained from a
generalized transition probability by `freezing' the $n-1$ first variables. Hence, given a generalized transition probability $K_n$ (where $n > 2$) and a fixed $(n-1$)-tuple $(\y_0, \y_1,\cdots, \y_{n-2}) \in S^{n-1}$, we denote $\tilde{\y}_{n-2}= (\y_0, \y_1,\cdots, \y_{n-2})$ and denote the transition probability
$K_{n, \tilde{\y}_{n-2}}$ by
\beq
K_{n, \tilde{\y}_{n-2}}(\x;A)=K_n(\y_0,\cdots,\y_{n-2},\x;A).
\eeq

\begin{theorem}
\label{thm1}
Let $\pi$ be the density of a target distribution supported on a bounded measurable
subset $D\subset \mathbb{R}^d$ and assume that $\pi$ is bounded from above. Then the two-stage adaptive Metropolis algorithm (TSAM),
as described in Section \ref{TSAM} is ergodic, i.e., 
if $(X_n)$ be a sequence of random vectors generated by the TSAM, then for any bounded and measurable function $f : D \rightarrow R$, the equality 
\beq
\lim_{n\rightarrow \infty}\frac{1}{n+1}\sum_{i=0}^{n}f(X_i) = \int_D f(\x)\pi(\x)d\x
\eeq
holds almost surely.
\end{theorem}
% Note that,  the auxiliary constants $c_0$ and $c_1$ depend on $D, \epsilon$ or $C_0$, and their actual value is irrelevant for our purposes here.
\begin{proof}
Using Theorem 2 from page 230 of \cite{haario2001}, it is sufficient to show that the sequence of generalized transition probability $(K_n)$, as defined in \eqref{tp}, satisfies the following three conditions:
\begin{enumerate}
    \item[(i)] There is a fixed integer $k_0$ and a constant $\lambda \in (0, 1)$ such that
    \beq
    \delta((K_{n, \tilde{\y}_{n-2}})^{k_o})\leq \lambda < 1 \ \text{for all } \ \tilde{\y}_{n-2} \in D^{n-1} \ \text{and} \ n\geq 2,
    \eeq
    where $\delta(\cdot)$ is the Dobrushin ergodicity coefficient \citep{Dobrushin_1956}.
    \item[(ii)] There is a probability measure $\pi$ on D and a constant $c_0 > 0$ such that
    \beq
    || \pi K_{n, \tilde{\y}_{n-2}} - \pi|| \leq \frac{c_0}{n} \ \text{for all } \ \tilde{\y}_{n-2} \in D^{n-1} \ \text{and} \ n\geq 2.
    \eeq
    
    \item[(iii)] The estimate for the operator norm
    \beq
    || K_{n, \tilde{\y}_{n-2}} - K_{n+k, \tilde{\y}_{n+k-2}} ||_{\mathcal{M}(D)\rightarrow \mathcal{M}(D)} \ \leq \ c_1\frac{k}{n}
    \eeq
    hold, where $c_1$ is a positive constant, $n, k > 1$ and it is assumed that $(n+k-1)$-tuple $\tilde{\y}_{n+k-2}$ is a direct continuation of $(n-1)$-tuple $\tilde{\y}_{n-2}$.
    \end{enumerate}
    
    Let us first verify that condition (i) is satisfied. From the definition of $C_t$ in  \eqref{ct} and the fact that D is bounded, it is trivial to show that all the covariances $C=C_n(\y_0, \y_1, \cdots, \y_{n-1})$ satisfy the matrix inequality  
    \beq
    \label{ct_in}
    0<a_1I_d \leq C \leq a_2I_d,
    \eeq
    where  $0 < a_1 < a_2 < 1$ are constants that depend on $D, \epsilon$ and $C_0$ and $I_d$ is the $d$-dimensional identity matrix. The inequality in 
    \eqref{ct_in} implies that the corresponding normal proposal densities $q_C(\x,\cdot)$ are uniformly bounded from below on D for all $\x \in S$. Using this fact along with conditions \eqref{boundpi}, \eqref{transk} and \eqref{tp}, we obtain the following bound     \beq
    \label{nume}
    K_{n, \tilde{\y}_{n-2}}(\x;A)\geq a_3\pi(A) \ \ \text{for all} \ \x \in S \ \text{and} \ A\subset D,
    \eeq
    where $a_3>0$ is a constant. 
    Using a proof similar to page 122--123 of \cite{nummelin84} \eqref{nume} yields  $\delta(K_{n, \tilde{\y}_{n-2}}) \leq  1 - a_3$. It proves that condition (i) is satisfied with $k_0 =1$.
    
    Now let us verify that condition (iii) is satisfied. Note that, for all $n\geq 2$ and for a given $\tilde{\y}_{n+k-2}$, the following inequality holds:
    \beq
    \label{iiia}
    || K_{n, \tilde{\y}_{n-2}} - K_{n+k, \tilde{\y}_{n+k-2}} ||_{\mathcal{M}(D)\rightarrow \mathcal{M}(D)} \ \leq \ 2 \sup_{\y \in D, A \in \mathcal{B}(D)}|K_{n, \tilde{\y}_{n-2}}(\y;A) - K_{n+k, \tilde{\y}_{n+k-2}}(\y;A)|.
    \eeq
    Now, fix $\y \in D$ and $A \in \mathcal{B}(D)$ and let $R_1 = C_n(\y_0, \cdots, \y_{n-2},\y)$ and $R_2 = C_n(\y_0, \cdots, \y_{n+k-2},\y)$. Then using \eqref{transk} and \eqref{tp}, we obtain 
    %\ber
    \beq
\begin{split}
    \label{iiib}
    |K_{n, \tilde{\y}_{n-2}}(\y;A) - K_{n+k, \tilde{\y}_{n+k-2}}(\y;A)| & = |Q_{R_1}(\y;A) - Q_{R_2}(\y;A)| \nonumber \\
    &\leq \Big|\int_{\x \in A}(q_{R_1}(\y-\x) - q_{R_2}(\y-\x))\alpha_1(\y,\x)\alpha_2(\y,\x)d\x \Big| \nonumber \\
    &  \ \ \ + \Big| \chi_A(\y)\int_{\mathbb{R}^d}(q_{R_1}(\y-\x) - q_{R_2}(\y-\x))(1-\alpha_1(\y-\x))d\x \Big| \nonumber \\
    & \ \ \ + \Big|\chi_A(\y)\int_{\mathbb{R}^d}(q_{R_1}(\y-\x) - q_{R_2}(\y-\x))\alpha_1(\y,\x)(1-\alpha_2(\y,\x))d\x \Big| \nonumber \\
    &\leq 3\int_{\mathbb{R}^d}\big| q_{R_1} (\z)- q_{R_2}(\z)\big| d\z \nonumber \\
    &\leq 3\int_{\mathbb{R}^d} \left(\int_{0}^{1}\left|\frac{d }{ds}q_{R_1+s(R_2-R_1)}(\z)\right|ds \right)d\z . \nonumber 
    %\eer
    \end{split}
\eeq
    Using \eqref{ct_in} it can be shown that the partial derivatives of the density $q_{R_1+s(R_2-R_1)}$ with respect to the components of the covariance are integrable over $\mathbb{R}^d$. This also yields
    \beq
    \label{iiic}
     \int_{\mathbb{R}^d} \left(\int_{0}^{1}\left|\frac{d }{ds}q_{R_1+s(R_2-R_1)}(z)\right|ds \right)d\z \leq a_5||R_2 - R_1||,
    \eeq
    where $a_5$ depends only on $\epsilon, C_0$ and $D$.
    Now from the recursive formula for $C$ it is straightforward to show that $|C_t - C_{t+1}|\leq a_6/t$ for $t>1$. Using this inductively, we can obtain from  \eqref{ct_in}
    \beq
    \label{iiid}
    ||R_2 - R_1||\leq a_7(\epsilon, C_0, D)k/n.
    \eeq
    Combining \eqref{iiia}, \eqref{iiib}, \eqref{iiic}, and \eqref{iiid} yields (iii).
    
    Finally, we verify that condition (ii) is satisfied. 
    Let us fix $\tilde{\y}_{n-2} \in D^{n-1}$ and denote $C^* = C_{n-1}(\y_0,\cdots, \y_{n-2})$. Using the same argument as \eqref{iiid} it follows that $||C^* - C_{n}(\y_0,\cdots, \y_{n-2},\y) ||\leq a_8/n$, where $a_8$ does not depend on $\y \in D$. Now proceeding exactly same as \eqref{iiia}, \eqref{iiib}, and \eqref{iiic}, we  obtain 
    \beq
    \label{iia}
    || K_{n, \tilde{\y}_{n-2}} - Q_{C^*} ||_{\mathcal{M}(D)\rightarrow \mathcal{M}(D)} \leq a_9/n
    \eeq
    As $Q_{C^*}$ is the transition probability of the random-walk two-stage Metropolis-Hastings algorithm, we have $\pi Q_{C^*} = \pi$ (see proof in \citealp{Efendiev2006}).
    Hence, from \eqref{iia}, we obtain
    \beq
    || \pi K_{n, \tilde{\y}_{n-2}} - \pi||=|| \pi (K_{n, \tilde{\y}_{n-2}} -  Q_{C^*}) ||_{\mathcal{M}(D)\rightarrow \mathcal{M}(D)} \leq c_0/n.
    \eeq
    This completes the proof that condition (ii) is satisfied.
\qed
\end{proof}

\section{Applications of the two-stage adaptive Metropolis algorithm} \label{simulation}
\subsection{Simulation study from different multivariate distributions}
In this simulation study, we apply the TSAM algorithm to sample from different multivariate distributions. These are simple examples to illustrate that the TSAM algorithm is very similar to the AM algorithm in terms of ergodicity, mixing and convergence properties. The computing efficiency of the TSAM over the AM is not the focus of this simulation study as the probability density evaluations for these multivariate distributions are not expensive, and the scope of considering a significantly faster method to approximate these densities is very limited. The computing efficiency of the TSAM over the AM is shown in the next set of applications as described in subsections \ref{real1} and \ref{real2}.

\subsubsection{Simulation from a multivariate t-distribution} 
We first consider a $d$-dimensional multivariate shifted t-distribution with the location vector $\mathbf{\mu}$, shape matrix $\mathbf{\Sigma}$ and degrees of freedom $\nu$. The density of this target distribution is given by
\beq
\pi(\x|\mathbf{\mu},\Sigma,\nu) = \frac{\Gamma(\nu+d/2)}{\Gamma(\nu/2) (\nu \pi)^{d/2} \Sigma^{1/2}}\Big(1+\frac{(\x-\mathbf{\mu})^T\Sigma^{-1}(\x-\mathbf{\mu})}{\nu}\Big)^{-(d+\nu)/2}, \ \x \in \mathbb{R}^d.
\eeq

% Here we use the two-stage metropolis algorithm to sample from this target distribution. 
We use all the four methods MH, AM, TSMH, and TSAM to sample from this target distribution. In the first stage of the TSAM, we use the density of the multivariate normal distribution with mean $\mathbf{\mu}$ and covariance matrix $\mathbf{\Sigma}$ as the approximate target distribution $\pi^*(\cdot)$. In particular, here we consider a $8$-dimensional t distribution, i.e., $d=8$ with degrees of freedom $\nu = 10$, and the center at $\mathbf{\mu}=(0,1,2,3,4,5,6,7)$. The shape matrix $\mathbf{\Sigma}$ is such that the $(i,j)$-th element is given by $\sigma_i\sigma_j\rho^{|i-j|}$, where $(\sigma_1^2, \sigma_2^2,\sigma_3^2,\sigma_4^2, \sigma_5^2, \sigma_6^2, \sigma_7^2, \sigma_8^2)= (1,1,1,1,1,2,4,6)$, and $\rho=0.4$. In order to maintain a compact support for the TSAM, the distribution is truncated outside $5$ standard deviation in each component. The initial covariance $C_0$ for the TSAM and AM is taken to be $I_d 2.4^2/8$. Following \cite{Gelman1996} the random walk MH  and TSMH jump sizes are also taken in the same order $(2.4^2/8)$. The burn-in period was chosen to be half of the chain length for each of the four methods.
 
% In this simple example, there will not be significant gain in computation efficiency using the two-stage approach compared to the regular adaptive Metropolis algorithm. This is because both $\pi^*$ and $\pi$ have approximately the same runtime.The goal of this simulation study is to show that the two-stage adaptive Metropolis algorithm has the same ergdic property as that of the adaptive Metropolis algorithm. The computing efficiency gain will be studied in other examples as described in Sections \ref{real1} and \ref{real2}

To compare the ergodicity of the chain,  $m=100$ different sets of samples are generated using the algorithm for a given sample size $n$. For each set of samples, the Monte Carlo average of a bounded function $f=10 \exp(-0.1(x_1+x_2+\cdots + x_8))$ is computed.  Suppose for the $k$-th set, $X_1, X_2, \cdots, X_n$ are the samples generated by the algorithm. The Monte Carlo estimate for $\int f(\x)\pi(\x)d\x$ is given by the sample average from this $k$-th set; $\bar{f}_k = \frac{1}{n+1}\sum_{i=0}^{n}f(X_i)$, $k=1,2,\cdots, m$. This process is repeated for different sample sizes $n=500, \ 1,000, \ 2,000, \ 5,000, \ 10,000$. The mean and the standard deviation of $m$ Monte Carlo averages are plotted for each of the sample sizes using the four methods (TSAM, AM, MH, and TSMH) in Figure \ref{fig:boxplot}. 
	\begin{figure}[t!]
		\centering%
			\includegraphics[height=10cm, width=15cm]{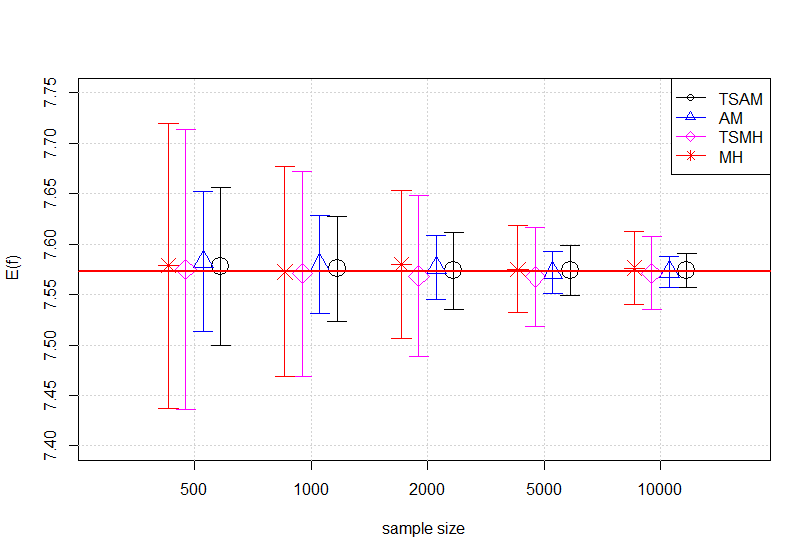}\
		\caption{The mean and one standard deviation error bar for Monte Carlo estimates of $E(f)$ using the four sampling methods. The red horizontal line denotes the value of the true $E(f)$. The target distribution is the $8$-dimensional multivariate shifted t-distribution.}
		\label{fig:boxplot}
	\end{figure}
The true value of the integral $E(f)=\int(f(x)\pi(x)dx$ is computed by simulating a random sample of size 1,000,000 using the {\it rmvt} function of the R package {\it mvtnorm}. The true value of $E(f)$ is shown in the red horizontal line. We observe that as the sample size increases the Monte Carlo average converges to the true $E(f)$ for all the four methods, confirming the ergodicity of the corresponding chains. 
Moreover, the results expressed in the figure indicate that the AM and TSAM algorithms simulate the target distribution more accurately than the MH and TSMH algorithms. 
	 \begin{figure}[h!]
		\centering%
		\subfloat[Along first principal component]{{\includegraphics[width=7cm, height=6cm]{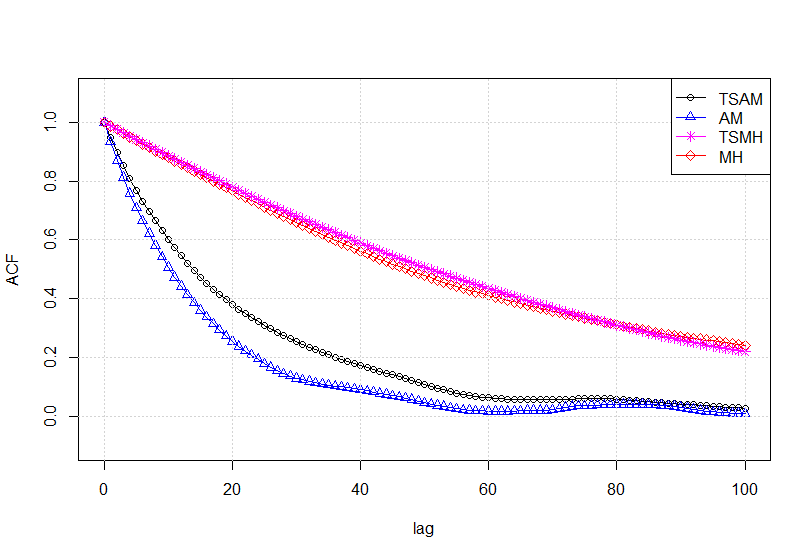}}}%
	    \qquad
		\subfloat[Along orthogonal direction]{{\includegraphics[width=7cm, height=6cm]{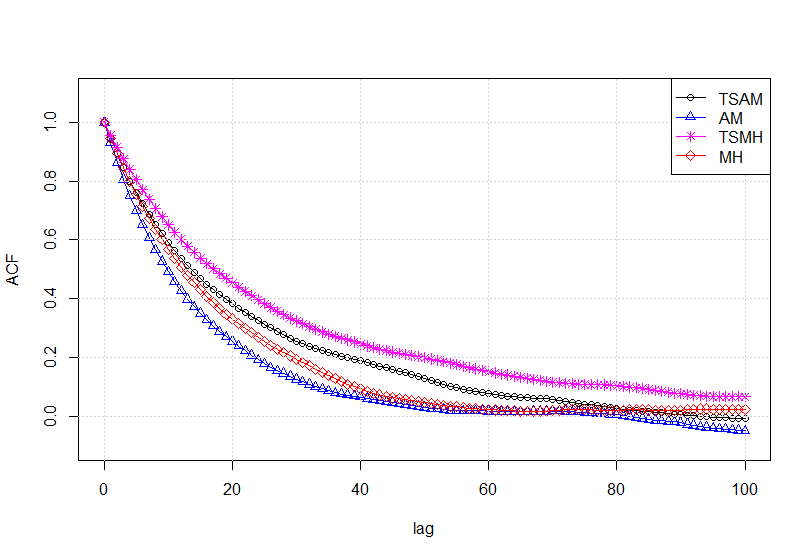} }}%
		\caption{The autocorrelation plot of different lags using the four sampling methods on the multivariate shifted t-distribution.}
		\label{fig:t_acf}
	\end{figure}
In Figure \ref{fig:t_acf} the autocorrelations of different lags of the samples generated by the four algorithms are plotted for two projections. In the direction of the largest principal component of the target distribution, the autocorrelation of the MH and TSMH algorithm indicates weaker convergence. The corresponding autocorrelation plot for AM and TSAM algorithms indicates much stronger converges of the corresponding chains and hence produces higher effective sample sizes. Note that both AM and TSAM perform very similarly in this regard.

\subsubsection{Simulation from a banana-shaped multivariate normal distribution}
%Example2: 
In the next example, we consider a $d$-dimensional ``banana-shaped" multivariate normal distribution where some components are non-linearly dependent. The non-linear dependency creates banana-shaped contours for the corresponding components. 
The non-linear banana-shaped distributions are constructed from the Gaussian ones by `twisting' them as follows. Let $\pi_1$ be the density of a multivariate normal distribution $N (\boldsymbol{\mu},\mathbf{\Sigma})$, then the density function of the ``twisted" Gaussian with non-linear parameters $a$ and $b$ is given by
%\beq
$\pi = \pi_1 \circ \phi_{a,b}(\x)$,
%\eeq
where $\phi_{a,b}(\x)=(ax_1 , x_2/a + ba^2(x_1^2+1), x_3, \cdots, x_d)$. Thus $\phi_{a,b}$ only changes the first and second coordinate and the determinant of the Jacobian of $\phi_{a,b}$  is identically equal to 1. This makes it very easy to calculate the confidence regions for the twisted Gaussian targets. We apply all the four methods (MH, AM, TSMH, and TSAM) to sample from this target distribution. In the first stage of TSAM, the density of the multivariate normal distribution with mean $\boldsymbol{\mu}$ and covariance matrix $\mathbf{\Sigma}$ is used as the approximate target distribution $\pi^*$ and  in the second stage, the true density $\pi$ is used. In particular, here we consider a $8$ dimensional ``banana-shaped" multivariate normal distribution where the ``twisting"  parameters for the first two components are taken as $ a= 1, b=0.05$. The mean $\boldsymbol{\mu}$ is  a zero vector and the covariance $\mathbf{\Sigma}$ is a diagonal matrix whose first diagonal element is $10$ and the rest are all $1$. To maintain the compact support of the target distribution for the TSAM we truncate the distribution outside 5 standard deviation in each component. The initial covariance $C_0$ for the TSAM and AM is taken to be $I_d 2.4^2/8$, which is also the jump-size for the random walk MH and TSMH algorithm. The burn-in period was the half of the chain length for all the four methods.
\begin{figure}[t!]
		\centering%
			\includegraphics[height=10cm, width=15cm]{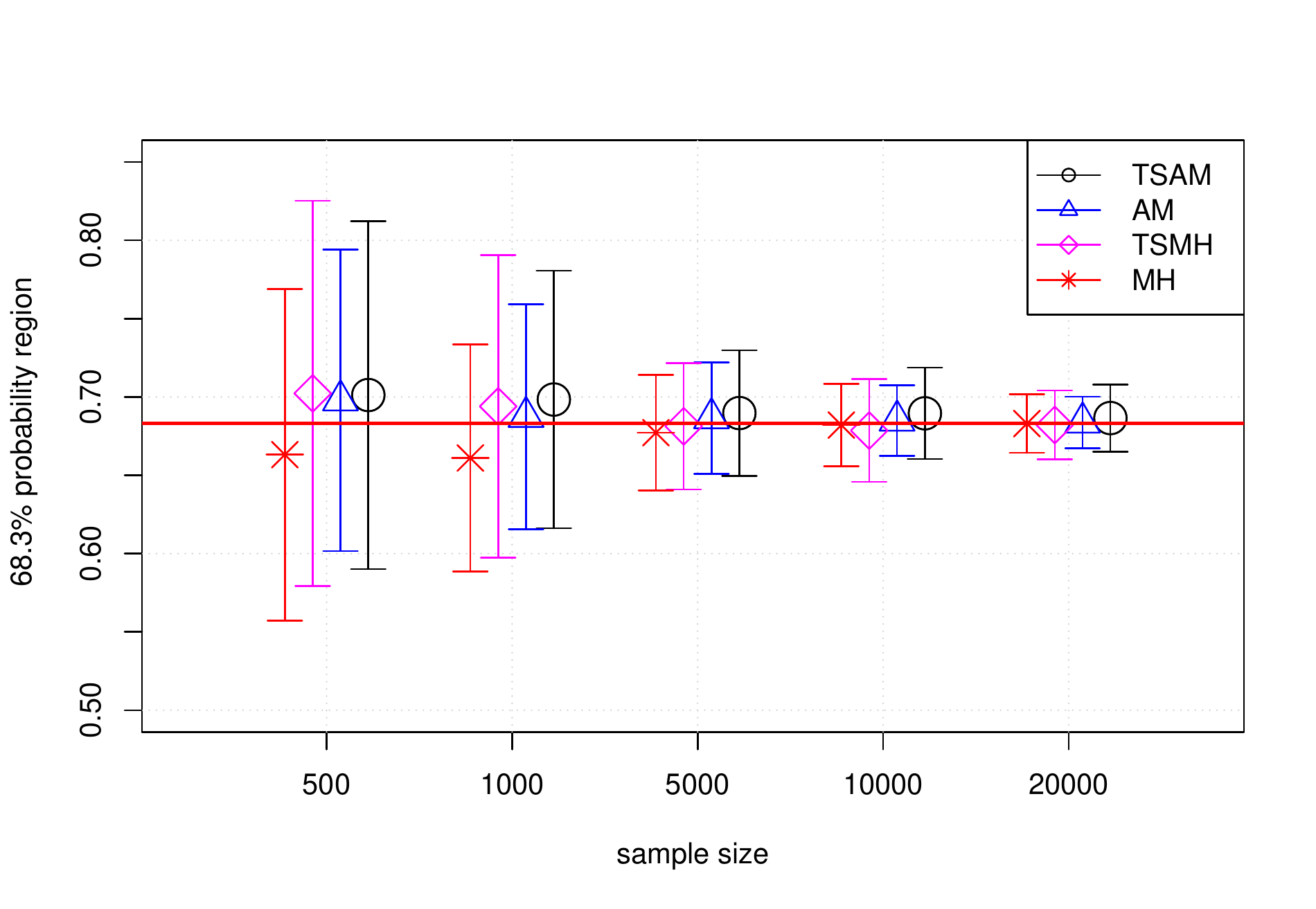}\
		\caption{The mean and one standard deviation error bar for the Monte Carlo estimates of 68.3\% probability region using the four methods. The red line denotes the true value. The target distribution is the banana shaped transformed multivariate normal distribution.}
		\label{fig:banana_ergo}
\end{figure}
% The quantity of interest is the $68.3$\% probability region for $\pi$ around its mean. All the four methods MH, TSMH, AM and TSAM were run $100$ times each with sample size

To compare the ergodicity of the chains,  $m=100$ different sets samples of size $n$ are generated using the four algorithms. For each set of samples, the Monte Carlo estimate of the $68.3$\% coverage probability is computed by the proportion of samples that falls within the corresponding region. The process is then repeated for different sample sizes $n=500$,  1,000,  2,000,  5,000,  10,000,  20,000. The mean and the standard deviation of $m$ Monte Carlo estimates are plotted for each of the sample sizes using the four methods in Figure \ref{fig:banana_ergo}. The true value $0.683$ is shown by the horizontal red line. We observe that Monte Carlo estimates converge to the true value as sample size increases for all the four methods. We also observe that the standard deviation is smaller for the AM and TSAM than MH and TSMH respectively, confirming that the adaptive algorithms more accurately  sample the target distribution.
\begin{figure}[t!]
		\centering%
		\subfloat[Two-stage adaptive Metropolis]{{\includegraphics[width=7.5cm, height=6cm]{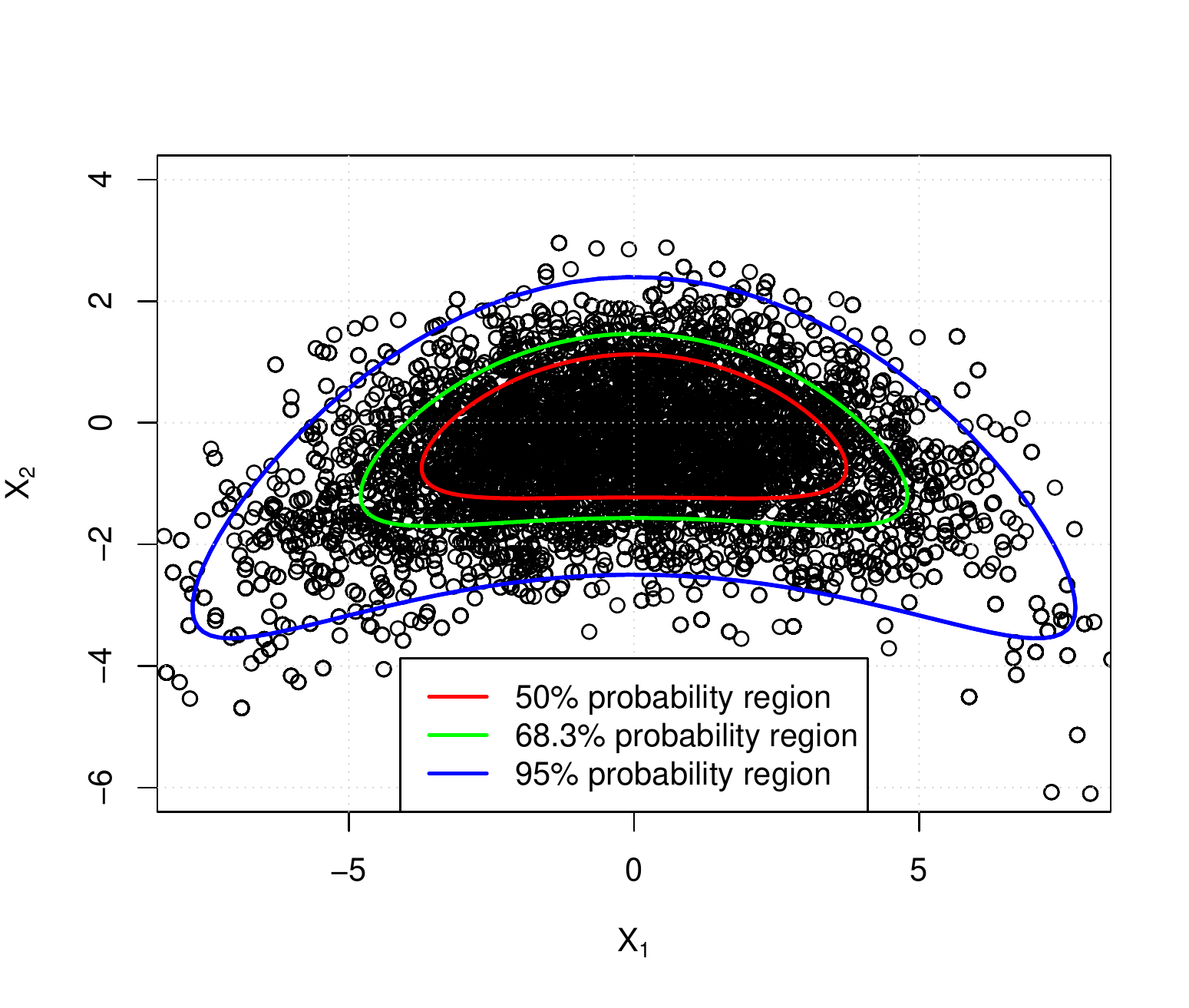}}}%
	    \qquad
		\subfloat[Adaptive Metropolis]{{\includegraphics[width=7.5cm, height=6cm]{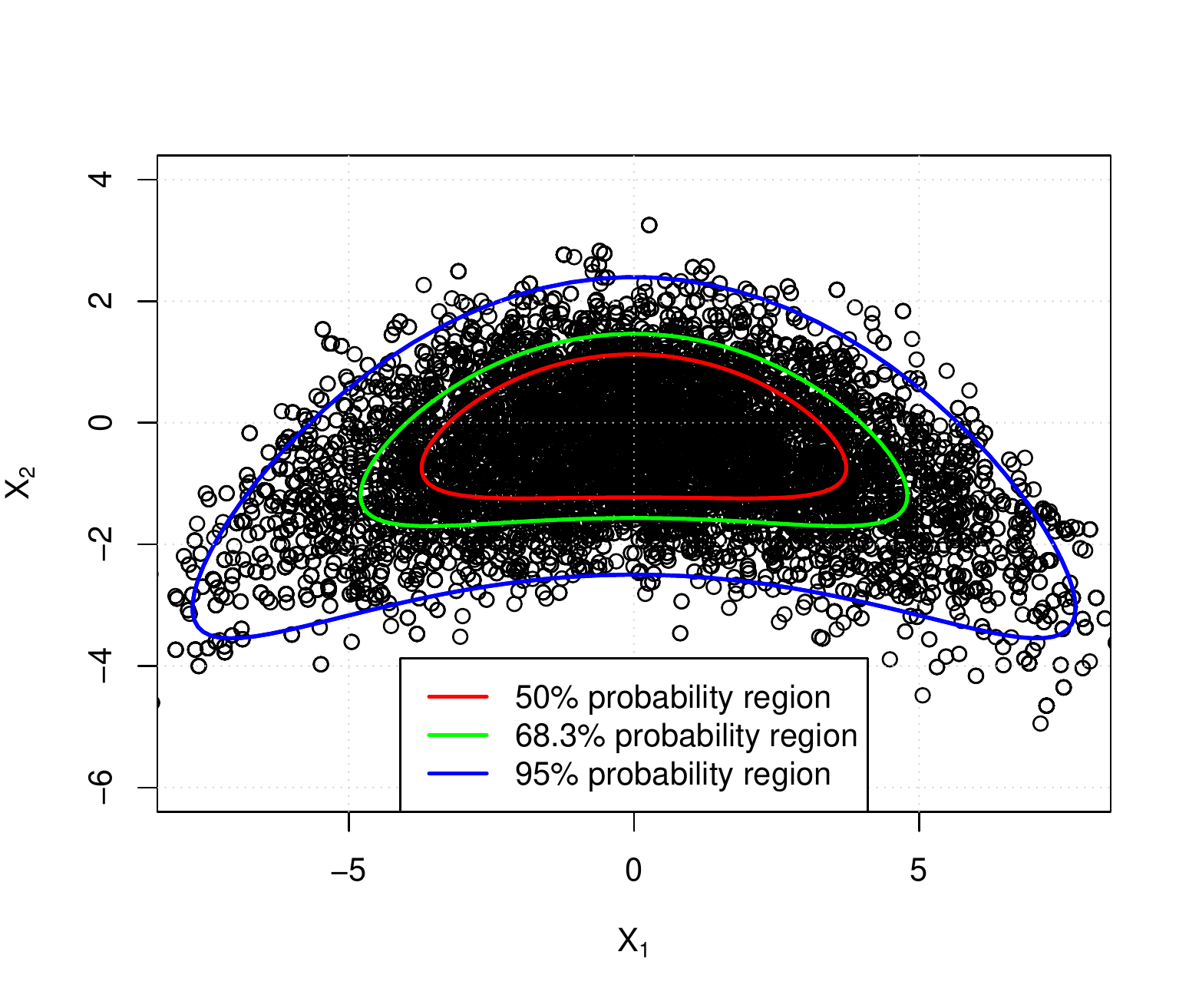} }}%
		\caption{Two-dimensional scatter plots of the samples from banana shaped transformed multivariate normal distribution and the corresponding theoretical contours.}
		\label{fig:banana_con}
\end{figure}
The $50$\%, $68.3$\%, and $95$\%	two-dimensional theoretical contours for the two `twisted' axes are plotted in Figure \ref{fig:banana_con} on top of the scatter plot of 20,000 samples from AM and TSAM, respectively. We  observe that the scatters of the samples from the two algorithms are very similar, and they are in agreement with the theoretical contours.
\begin{figure}[t!]
	\centering%
	\subfloat[Along first principal component]{{\includegraphics[width=7cm, height=6cm]{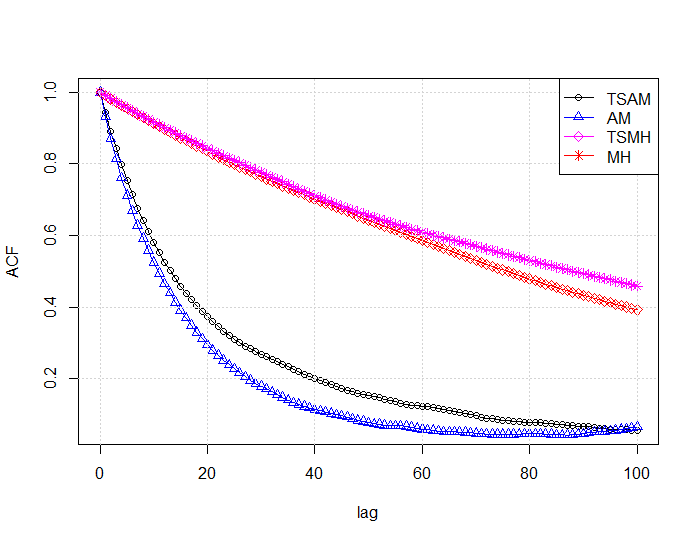}}}%
    \qquad
	\subfloat[Along orthogonal direction]{{\includegraphics[width=7cm, height=6cm]{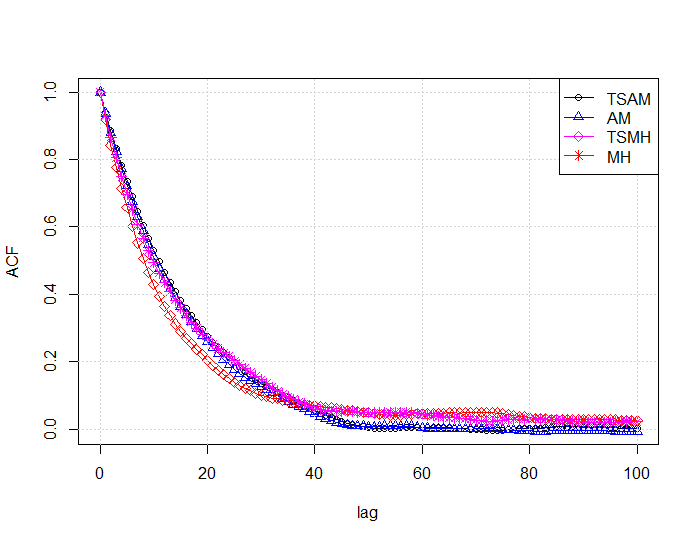} }}%
	\caption{The autocorrelation plot of different lags  from samples generated by the four sampling methods. The target distribution is the banana shaped transformed multivariate normal distribution.}
	\label{fig:banana_acf}
\end{figure}
\begin{figure}[h!]
\centering%
	\includegraphics[height=10cm, width=16cm]{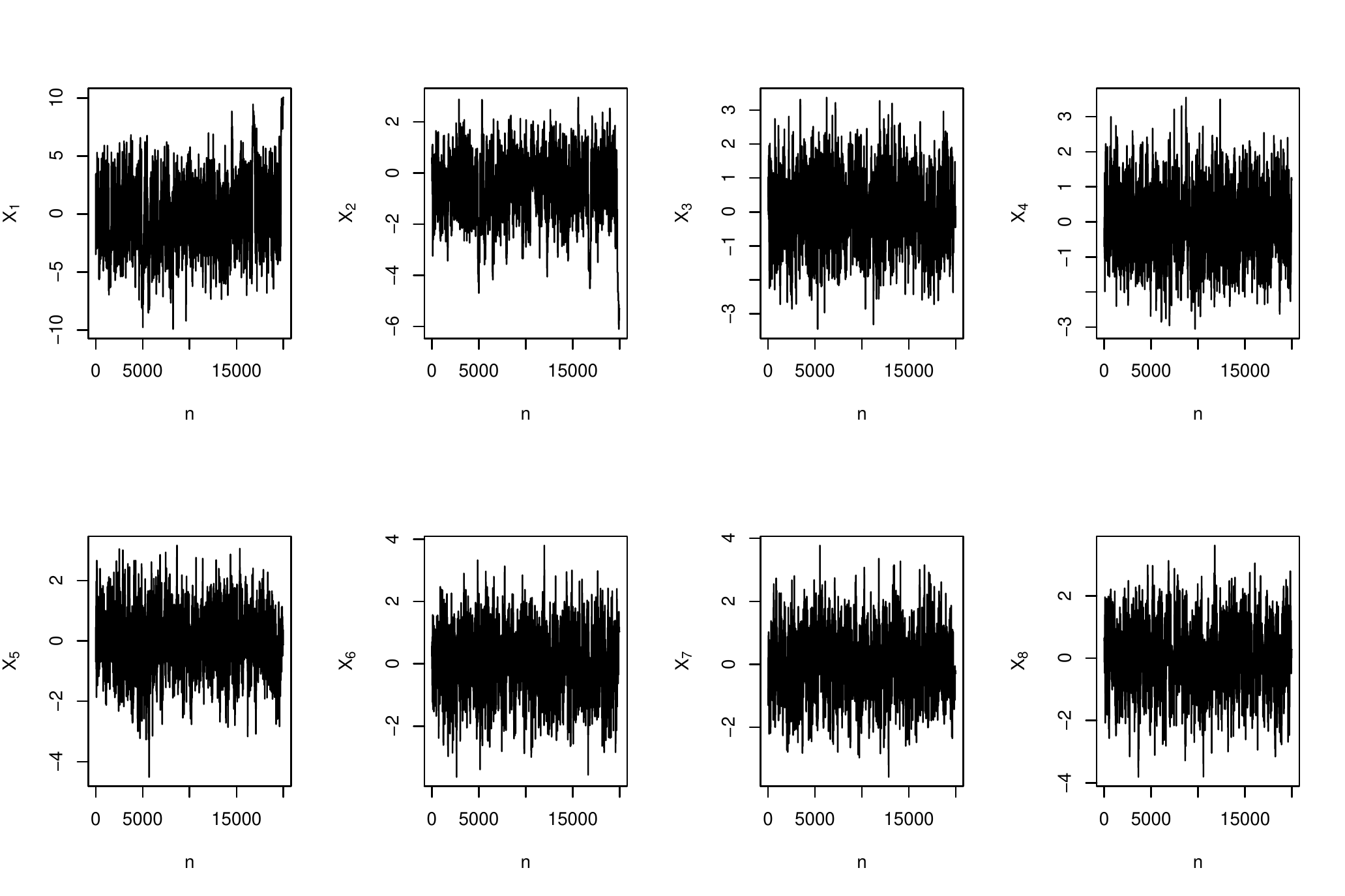}\
	\caption{The trace-plots of samples from the banana shaped multivariate normal distribution using TSAM.}
\label{fig:chain1}
\end{figure}
In Figure \ref{fig:banana_acf} the autocorrelations of different lags are plotted along two projections -- the direction of the largest principal component and the direction of an orthogonal component. The results are very similar to the previous example, i.e., AM and TSAM have much stronger convergence of the corresponding chains than MH and TSMH. Thus, the effective sample sizes for AM and TSAM methods are much larger than MH and TSMH. We also observe that both AM and TSAM perform very similarly in this regard. 
Finally, Figure \ref{fig:chain1} represents the trace-plot of all $8$ components of the samples from $\pi$ using TSAM. The trace-plot indicates that the chain has converged and mixed very well.

\subsection{Application in Bayesian logistic regression for tall data}
\label{real1}

Here we consider a binary classification problem using logistic regression for tall data to compare the computational efficiency of the TSAM over the traditional AM. Tall data here means that the number of observations is very large. The dataset is based on a phone marketing campaign conducted by a Portuguese bank to 41,188 potential customers ( see \cite{moro2014data} for details).  Let $y$ be the binary response variable, which represents whether or not a customer subscribed to a term deposit after contact. The predictor variables ($\bm{x}$) are the number of employees of the bank (four categories), type of job (three categories), type of contact, the month of the promotion call (three categories), and previous campaign outcome (three categories). 
We use the classical Bayesian logistic regression model 
\beq
y_i | \bm{\beta}, \bm{x}_i \sim \text{Bernoulli} (\pi (\bm{x}_i)),  \ \
%\end{equation}
%\begin{equation
\pi (\bm{x}_i) = 1/[1+\exp(-\bm{x}_i^T\bm{\beta})], \mbox{ for } i = 1,2, \cdots, N.
%\end{equation}
%\begin{equation}
%\bm{\beta} \sim MVN(\bm{0}, \mathbf{\Sigma}_0),
\eeq
We assume a vague prior for the  $d$-dimensional regression parameter  $\bm{\beta}$ that is  $\bm{\beta} \sim MVN(\bm{0}, \mathbf{\Sigma}_0)$, where $\mathbf{\Sigma}_0 = 100 I $. The posterior distribution of the model parameters can be expressed as $p(\bm{\beta} |\bm{y} , \bm{x}) \propto p(\bm{y}|\bm{\beta}, \bm{x})p(\bm{\beta})$.
Markov chain Monte Carlo methods are generally used for a sampling-based inference from this posterior distribution. 
As the number of observations (41,188) is very large the computation of the full likelihood is expensive. The full log-likelihood for the logistic regression model as also be written as two sums: 
\begin{equation}
\label{full_like}
l(\bm{\beta})=\log\left(p(\bm{y}|\bm{\beta}, \bm{x})\right) = \sum_{i:y_i=1} \{\bm{x}_i^T\bm{\beta} - \log(1+\exp(\bm{x}_i^T\bm{\beta}))\} +\sum_{j:y_j=0} \{ - \log(1+\exp(\bm{x}_i^T\bm{\beta}))\}.
\end{equation}
%where $\theta_i = \bm{x}_i \bm{\beta}$.
As the data is imbalanced (too many zeros than ones), the computation of the first sum is relatively cheap and the computation of the second term is expensive. \cite{raftery2012fast} approximated the second term using a sub-sampling method. The corresponding approximated log-likelihood is given by
\begin{equation}
\label{app_like}
l^*(\bm{\beta}) = \sum_{i:y_i=1} \{\bm{x}_i^T\bm{\beta} - \log(1+\exp(\bm{x}_i^T\bm{\beta}))\} +\frac{N_0}{n_0}\sum_{j \in S} \{ - \log(1+\exp(\bm{x}_i^T\bm{\beta}))\},
\end{equation}
where $N_0$ is the number of zero responses for the full data and $n_0 (<<N_0)$ is the number of zero responses in a sub-sample $S$. The sub-sample is selected randomly from the entire sample. As this approximated likelihood only needs a scan of the partial data, it is significantly cheaper to compute than the full likelihood.
% It is computationally expensive to calculate the likelihood (4.5) on each step of the MCMC algorithm. In the TSAM method, the inexpensive case-control approximate likelihood \citep{raftery2012fast} is used in the first stage which is described below.
Th approximated likelihood \eqref{app_like} is used to screen out bad proposals in the first stage of TSAM. The expensive likelihood \eqref{full_like} will only be evaluated at the second stage if the proposal is accepted in the first stage. 
So, we have
$\pi^*(\bm{\beta}) \propto \exp(l^*(\bm{\beta}))p(\bm{\beta})$ and $\pi(\bm{\beta}) \propto \exp(l(\bm{\beta}))p(\bm{\beta})$ for the TSAM algorithm given in Section \ref{TSAM}.%

To evaluate the efficiency of the MCMC algorithms, accounting for autocorrelation, the effective draws per minute (EDPM) is a good measure \citep{payne2018}. The EDPM is a measure of the equivalent number of independent posterior draws per minute the MCMC chain represents and is defined as 
\begin{equation}
    \displaystyle
    EDPM = t^{-1} \left( \frac{n}{1+2\sum_{k=1}^{\infty} \rho_k} \right).
\end{equation}
The EDPM incorporates both the execution time and autocorrelation of the chain to measure its efficiency. Here $\rho_k$ is the autocorrelation of lag $k$ for the posterior samples and $t$ in the total CPU time (in minutes) to generate $n$ samples from the target distribution. Similarly, the relative effective draws per minute (REDPM) of MCMC sampler $1$ over MCMC sampler $2$ is given by $\frac{EDPM_1}{EDPM_2}$. Here we use the REDPM to compare the overall performance of TSAM and AM.
	\begin{figure}[t!]
		\centering%
			\includegraphics[height=12cm, width=16cm]{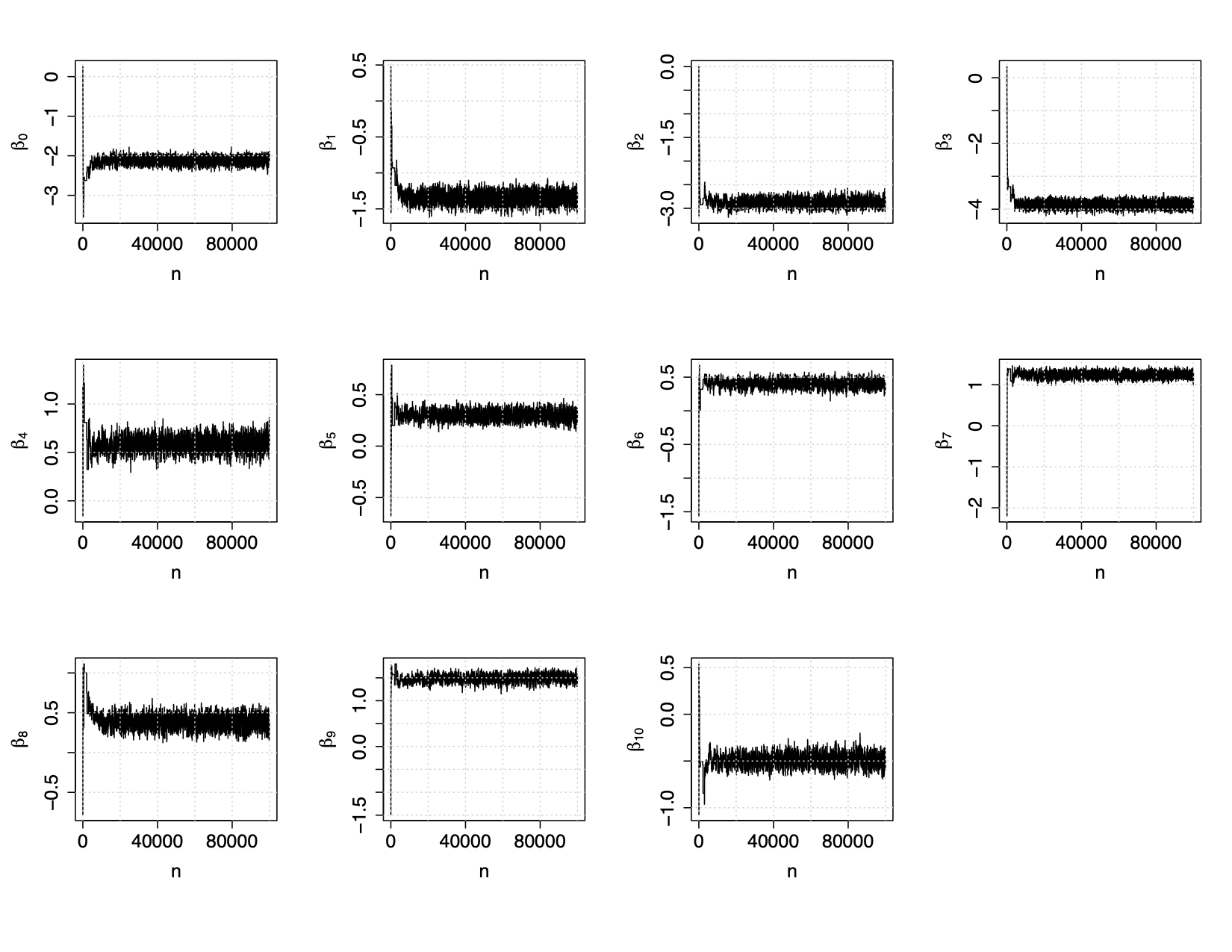}\
		\caption{The trace-plot of posterior samples generated by TSAM. Each plot is for one of the 11 parameters of the logistic regression model.}
		\label{fig:chain_bank}
	\end{figure}
		
In our example we have a total of 41,188 observations, out of which 36,548 observations have zero responses. We randomly subsampled 10,000 of these to compute the approximated likelihood, i.e., $N_0=$36,548 and $n_0=$10,000. As a result, the approximated likelihood used in the first stage of TSAM is about 4 times faster than the true likelihood.  Figure \ref{fig:chain_bank} demonstrates that the convergence and the mixing of the chains are good for the TSAM algorithm. The marginal posterior distributions of the model parameters from the TSAM and AM are shown in Figure \ref{fig:post_bank}. We observe that the posteriors from the two methods overlap with each other, confirming that both methods behave similarly in terms of approximation of the true posterior distribution. The overlapping autocorrelation plots in Figure \ref{fig:redpm_bank} (a) also confirm the similar mixing and convergence properties of the two methods.
	\begin{figure}[h!]
		\centering%
			\includegraphics[height=10cm, width=16cm]{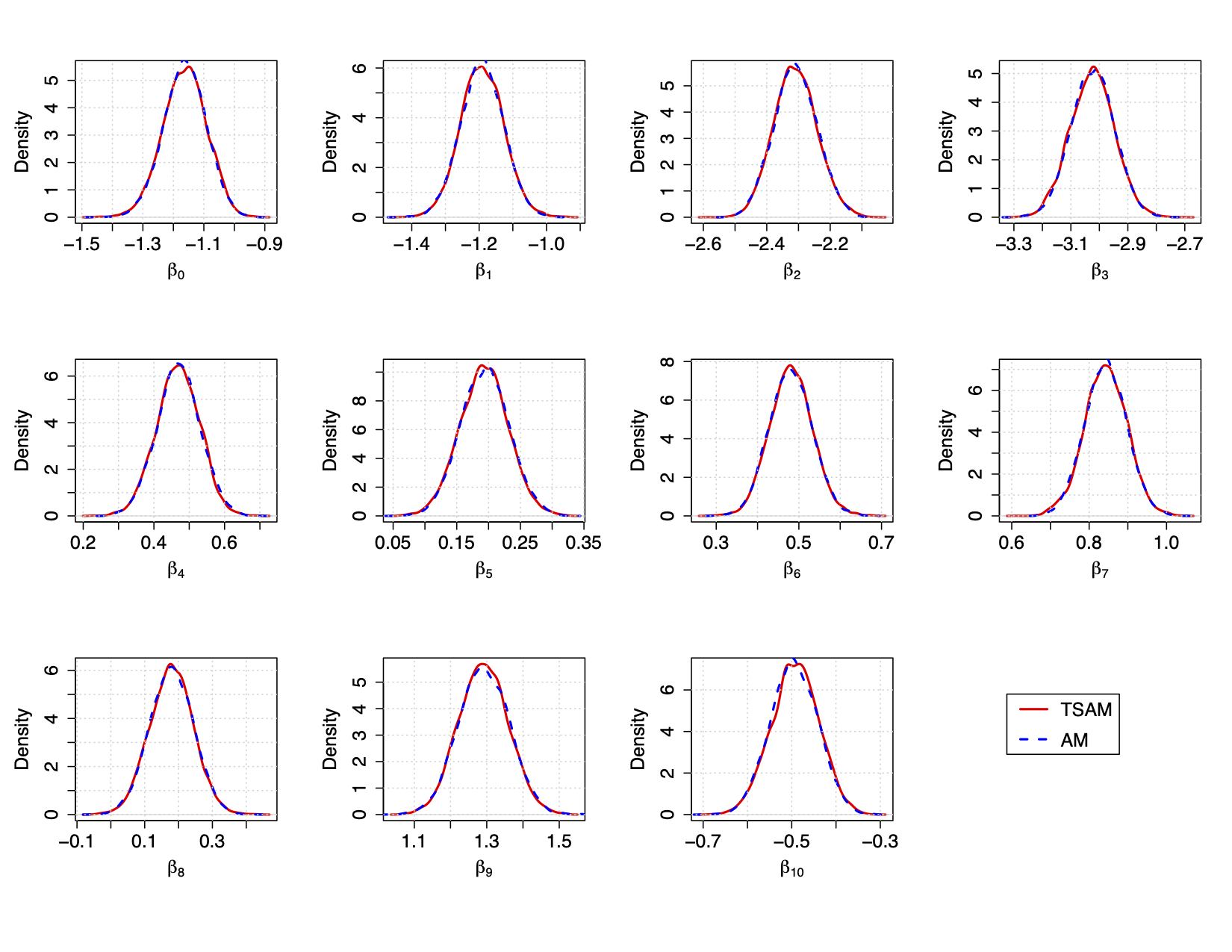}\
		\caption{The marginal posterior density for the 11 parameters of the logistic regression model. The posterior samples are generated using the TSAM and the AM method.}
		\label{fig:post_bank}
	\end{figure}
	 \begin{figure}[h!]
		\centering%
		\subfloat[Log Posterior]{{\includegraphics[width=7.5cm, height=6cm]{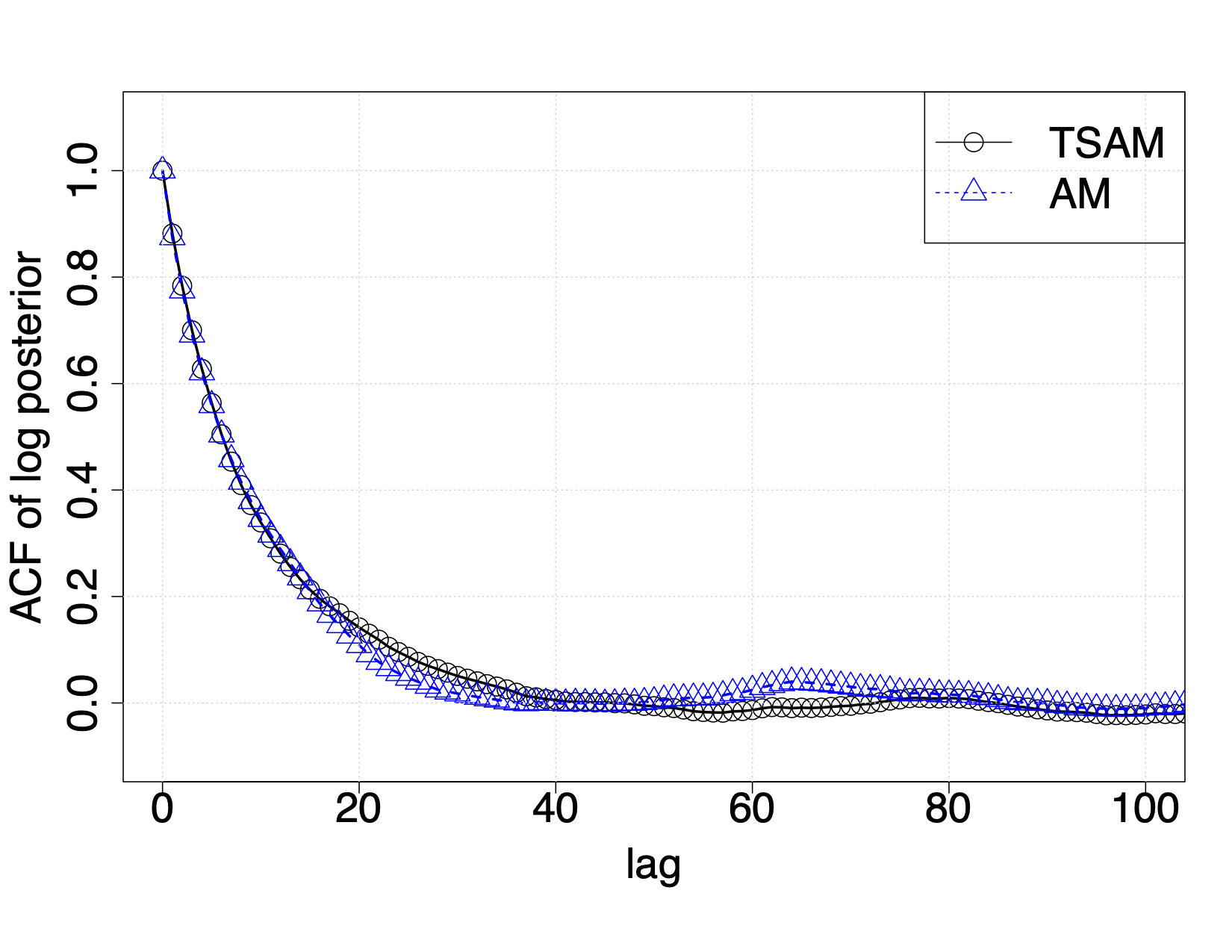}}}%
	    \qquad
		\subfloat[REDPM]{{\includegraphics[width=7.5cm, height=6cm]{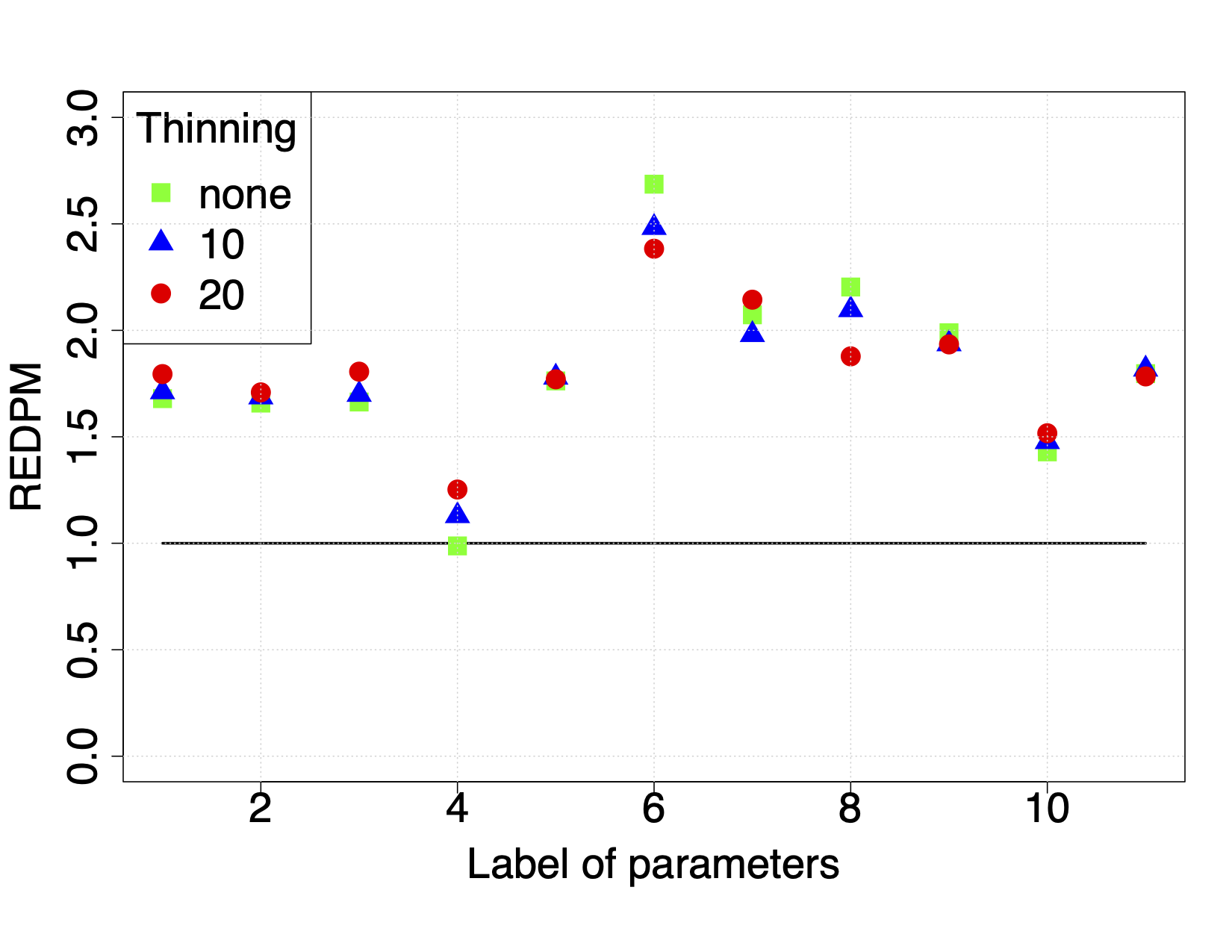} }}		
		\caption{(a): The autocorrelation plot of different lags for the log-posterior density in the logistic regression model; (b): REDPM of TSAM over AM for all the 11 model parameters and three thinning strategies.}
		\label{fig:redpm_bank}
	\end{figure}
The REDPM of the TSAM method over the AM method for each coefficient $\beta_0, \beta_1, \cdots, \beta_{10}$ are shown in Figure \ref{fig:redpm_bank} (b). Here three cases were considered i) no thinning: all the samples after the burn-in are used to calculate the REDPM ii) thinning 10: every 10-th sample from the chain is used to calculate the REDPM and iii) thinning 20: every 20-th sample from the chain is used to calculate the REDPM. We observe that the REDPM of TSAM over AM are all above one for each of the three cases and for each parameter. The REDPM of TSAM over AM using the autocorrelations of the log posterior is $1.53$. So we can conclude that the TSAM  algorithm performs better than AM in terms of computation efficiency while the accuracy of the two methods is very similar.
    % tsam chain plo while 
%
    % log posterior plot 
% 	\begin{figure}[h]
% 		\centering%
% 			\includegraphics[height=12cm, width=16cm]{autocor_pca_bank.png}\
% 		\caption{autocorrelation of the first and fourth components of TSAM and AM}
% 		\label{fig:chain_pred}
% 	\end{figure}
% 	% pca plots 
% 		\begin{figure}[h]
% 		\centering%
% 			\includegraphics[height=12cm, width=16cm]{autocor_logpost_bank.png}\
% 		\caption{autocorrelation of the log posterior of TSAM and AM}
% 		\label{fig:chain_pred}
% 	\end{figure}
	% tsam posterior over am posterior 
%
%	
% 	% redpm 
% 		\begin{figure}[h]
% 		\centering%
% 			\includegraphics[height=12cm, width=16cm]{REDPM_all_bank.png}\
% 		\caption{REDPM between TSAM and AM for all 11 parameters}
% 		\label{fig:chain_pred}
% 	\end{figure}
 We also tried the random walk MH and TSMH sampler for this posterior with jump size $2.4^2/11$, but the chain failed to converge within 100,000 iterations for both the methods.	

\subsection{Application in  Bayesian calibration of a predator-prey model} \label{real2}
% Many realistic inference problems involve complex likelihoods required to compute from an expensive numerical solver of a system of ODE or PDE, and correlated parameters to inference. For such a scenario, we can use a coarser numerical approximation of the solver in the first stage and a finer approximation in the second stage.   
In this example, we consider a Bayesian calibration problem in the Lotka-Volterra predator-prey system \citep{lotka1925principles, volterra1926fluctuations}.  To characterize the oscillating populations of predators and prey, Lotka and Volterra formulated the following parametric differential equations 
\beq
\frac{dy_1}{dt} = \alpha y_1 - \beta y_1 y_2, \ \ 
\frac{dy_2}{dt} = -\gamma y_2 + \delta y_1 y_2,
\label{LV}
\eeq
where $y_1$ represents the prey population, $y_2$ is the predator population, $\alpha$ is the prey growth rate, $\beta$ is the predation rate, $\gamma$ is the predator growth rate, and $\delta$ is the predator death rate. Given these model parameters and the initial population sizes, the forward model \eqref{LV} can be solved numerically to predict the predator and prey population. Here we consider the inverse problem of calibrating the model parameters given the observed data on the predator and prey population. In particular, we use the data of Canadian lynx and snowshoe hare populations between 1900 and 1920 collected annually by the Hudson’s Bay Company \citep{odum1953fundamentals} to calibrate the model parameters.
% likelihoods {\color{red} inter change notaton, "i and j", i stands for observations, j should be the component, be consistent with the previous example}
% Because the population states are always positive and their distributions are positively skewed, it is common to use log-transformation to them. 
Let us assume that $z_{ji}$ is the observed population of $j$-th species at time $t_i$, and $y_{ji}$ is the corresponding solution of the Lotka-Volterra equations \eqref{LV} using a numerical solver with initial values $(y_1^0, y_2^0)$, where $j = 1,2$, $i = 1,2, \cdots, N$ and $N$ is the number of total observations. Then, the  model for the observed data is given by  
\beq
\log(z_{ji}) = \log(y_{ji}) + \epsilon_{ji}, \ \
\epsilon_{ji} \iid N(0, \sigma^2_j),
\eeq
%The log population of prey and predator have error scales of $\sigma_1$ and $\sigma_2$. 
i.e., we assume
$ z_{ji} \sim {\rm LogNormal}(y_{ji}, \sigma_j^2) \ j,=1,2.$
We have a total of eight parameters to calibrate -- four parameters ($\alpha$, $\beta$, $\gamma$ and $\delta$) governing the system dynamics, two initial values ($y_1^0$ and $y_2^0$) for hare and lynx, and two parameters denoting the corresponding observational error variances ($\sigma^2_1$, $\sigma^2_2$). A Bayesian hierarchical model is used to calibrate these parameters of the model. More specifically, the solution to the inverse problem is given by the following posterior distribution of the model parameters given the data:
\beq 
    p(\alpha, \beta, \gamma, \delta, \sigma_1, \sigma_2, y_1^0, y_2^0|\z_1, \z_2) \propto p( \z_1,\z_2| \alpha, \beta, \gamma, \delta, \sigma_1, \sigma_2, y_1^0, y_2^0)p(\alpha)p( \beta)p( \gamma)p( \delta)p( \sigma_1)p( \sigma_2)p( y_1^0)p( y_2^0),
\eeq 
where $p( \z_1,\z_2| \alpha, \beta, \gamma, \delta, \sigma_1, \sigma_2, y_1^0, y_2^0) = \prod_{i=1}^{n}p(z_{1i}|y_{1i}, \sigma_2)p(z_{2i}|y_{2i}, \sigma_2)$ is the likelihood term, $y_{ji}$s are the solution to the system of ODEs \eqref{LV} and $p(z_{ji}|y_{ji}, \sigma_j)$ is the density of ${\rm LogNormal}(y_{ji}, \sigma_{j}^2), \ j=1,2; \ i = 1, 2, \cdots, N$.
% For the first stage, we solve the odes with finite difference method in monthly scaled grids, and daily scaled grids for the second stage. 
% priors 
The prior for all the parameters are assumed to be independent of each other. Weekly informative uniform priors are assigned  for the system dynamic parameters and are given by
\beq
\alpha, ~ \gamma \sim U(0, 0.1), \ \ 
\beta,~ \delta \sim U(0,0.01).
\eeq 
Log-normal priors are used for the noise standard deviations and the initial population sizes which are given by
\beq
    \sigma_j \sim {\rm LogNormal}(-1,1), \ \ y_j^0 \sim {\rm LogNormal}(\log(10), 1), \ j=1,2.
\eeq
% Log-normal distribution is also suitable for the initial population of prey and predator, which is positive constrained too. 
% \beq
%     z_i^0 \sim LogNormal(\log(10), 1) 
% \eeq 
% posterior 

We use this posterior distribution to illustrate the computing efficiency of the TSAM sampler over the usual AM sampler. Let us assume that for a given set of initial values and system parameters, the true solution $y_{ji}$ can be well approximated by solving equation \eqref{LV} using a finite difference method in  a very fine grid e.g., daily scale. This computation is relatively expensive, given the fact that we have to solve the system for 20 years and repeat it thousands of times for posterior sampling. On the other hand, if we solve the same equation \eqref{LV} using a finite difference method in a monthly scale, then it would be 30 times faster. Let $y_{ji}^*$ be the solution using the finite difference method in a monthly grid and $p^*(\alpha, \beta, \gamma, \delta, \sigma_1, \sigma_2, y_1^0, y_2^0|\z_1, \z_2)$ be the corresponding posterior where $y_{ji}$ is replaced by $y_{ji}^*$ in the likelihood.
This approximate posterior distribution is used in the first stage of TSAM and the original posterior is used in the second stage of TSAM. 

		\begin{figure}[t!]
		\centering%
			\includegraphics[height=12cm, width=16cm]{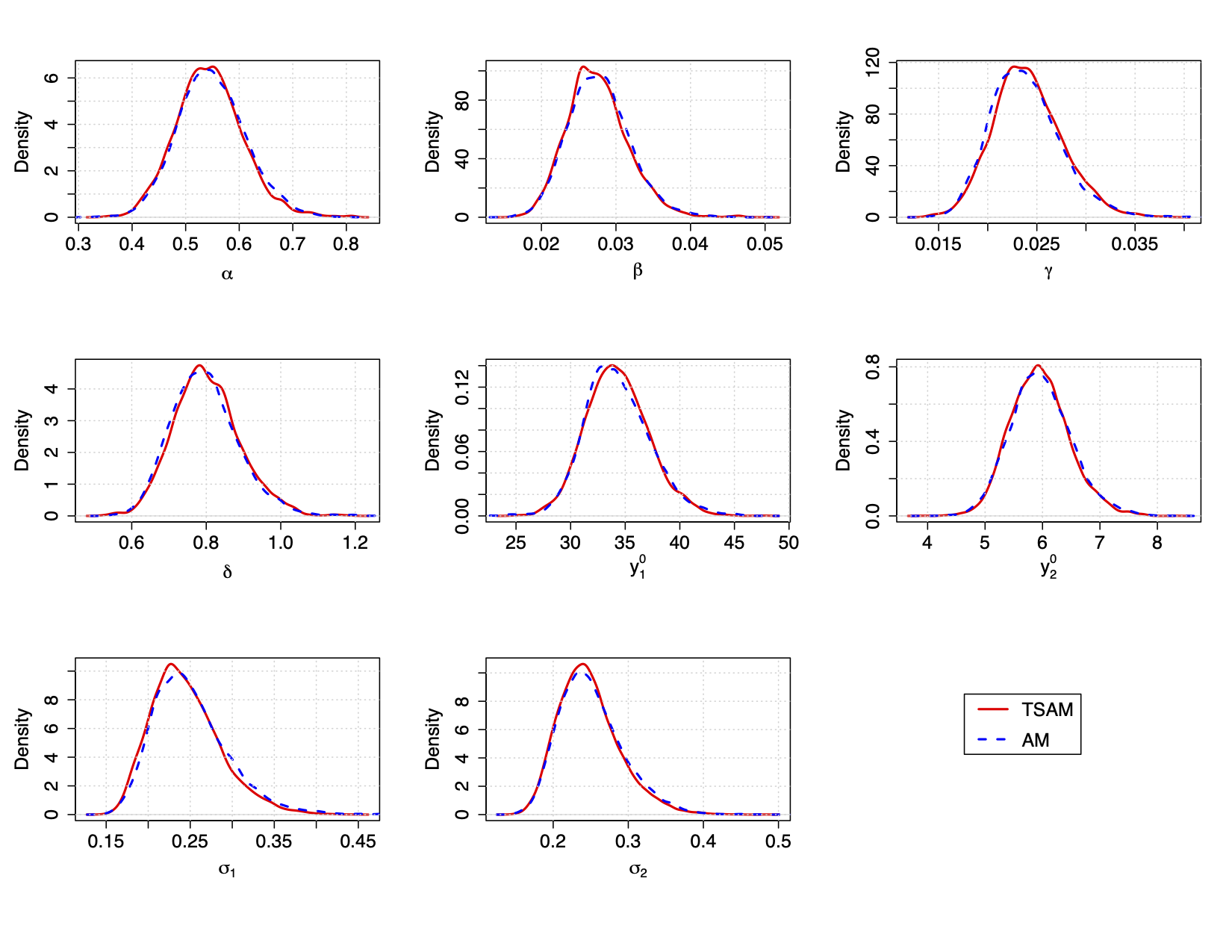}\
		\caption{The marginal posterior density of the 8 parameters in the predator-prey model. The posterior samples are generated using the TSAM and AM method. }
		\label{fig:post_pred}
	\end{figure}
%Figure \ref{fig:chain_pred} demonstrates that the chain of TSAM converged and mixed well.

	    % tsam chain plot 

% 	\begin{figure}[t]
% 		\centering%
% 			\includegraphics[height=12cm, width=16cm]{tsam_chain_predatorPrey.png}\
% 		\caption{The trace-plot of posterior samples generated by TSAM. Each plot is for one parameter of the predator-prey model.}
% 		\label{fig:chain_pred}
% 	\end{figure}
	
%     % pac posterior plot 
% 	\begin{figure}[h]
% 		\centering%
% 			\includegraphics[height=12cm, width=16cm]{autocor_pca_predatorPre.png}\
% 		\caption{autocorrelation of the first and fourth components of TSAM and AM }
% 		\label{fig:chain_pred}
% 	\end{figure}
	
% 	% log plots 
% 		\begin{figure}[h]
% 		\centering%
% 			\includegraphics[height=12cm, width=16cm]{autocor_logpost_predatorPre.pdf}\
% 		\caption{autocorrelation of the log posterior of TSAM and AM}
% 		\label{fig:chain_pred}
% 	\end{figure}
	
	% tsam posterior over am posterior 

% 		% redpm
% 		\begin{figure}[h]
% 		\centering%
% 			\includegraphics[height=12cm, width=16cm]{REDPM_all_predatorPre.pdf}\
% 		\caption{REDPM between TSAM and AM for 8 parameters }
% 		\label{fig:chain_pred}
% 	\end{figure}
	    \begin{figure}[t!]
		\centering%
		\subfloat[Log Posterior]{{\includegraphics[width=7.5cm, height=6cm]{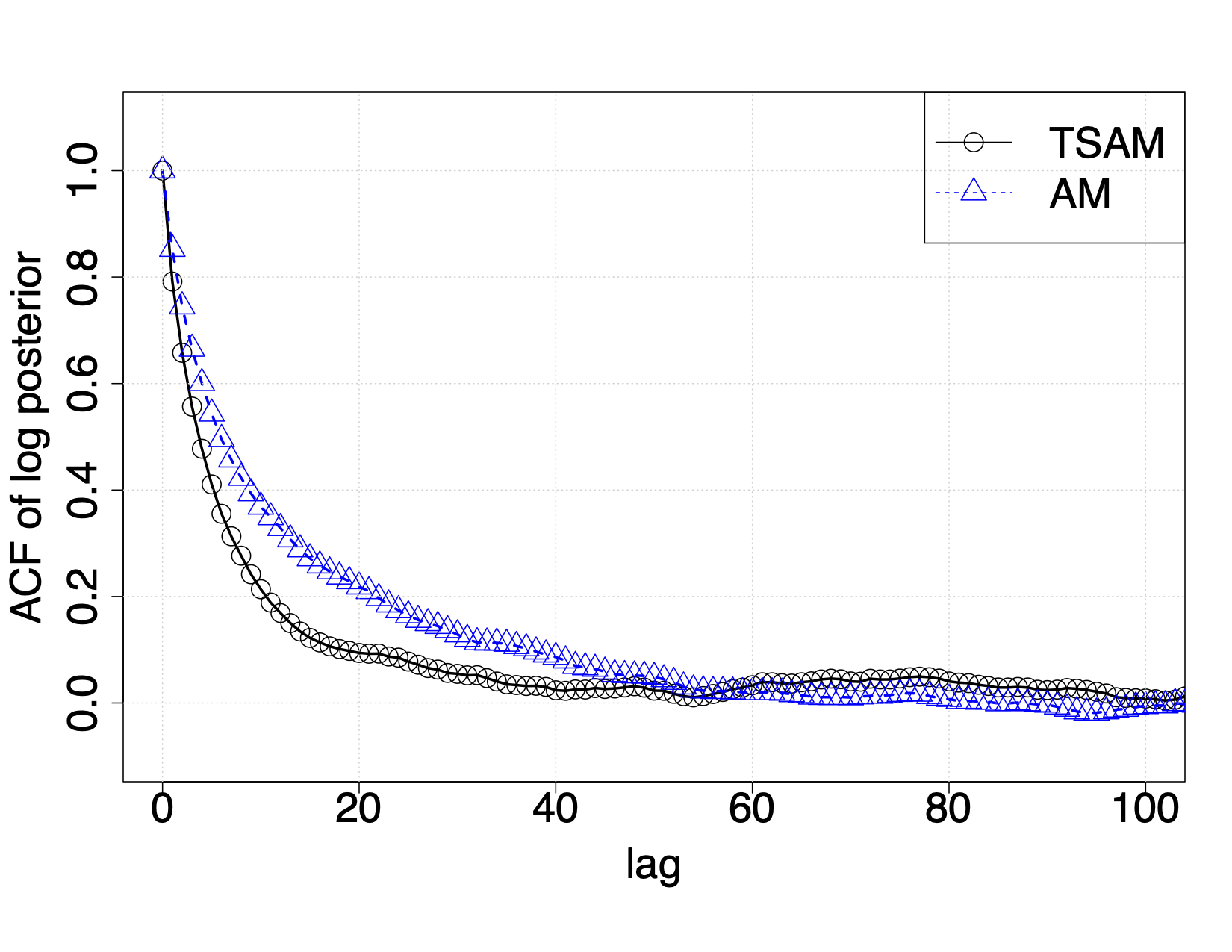}}}%
	    \qquad
		\subfloat[REDPM]{{\includegraphics[width=7.5cm, height=6cm]{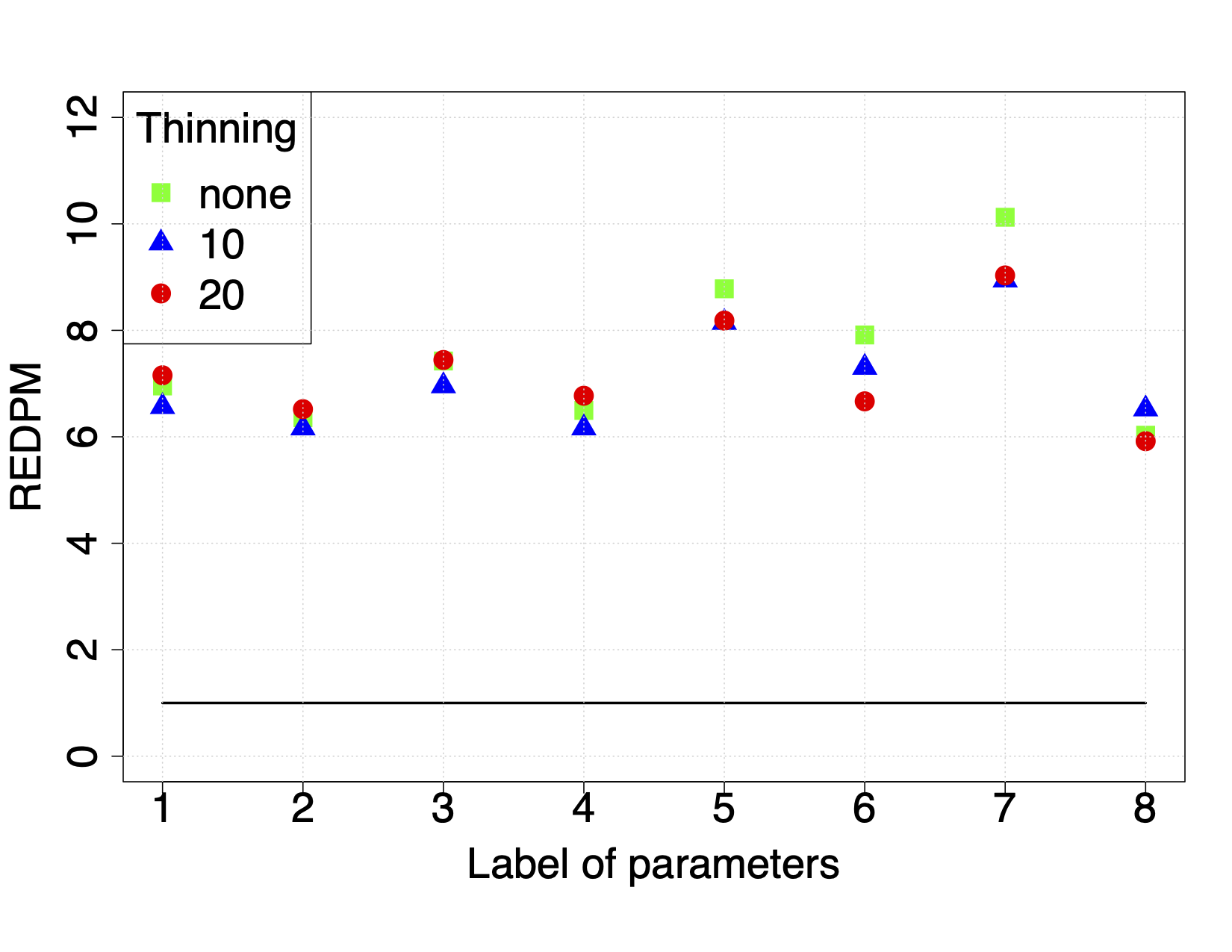} }}%
		\caption{(a): The autocorrelation plot of the log-posterior in the predator-prey example; (b): REDPM of the TSAM over the AM for all 8 parameters in the   predator-prey example and for the three thinning strategies.}
		\label{fig:redpm_predatorPrey}
	\end{figure}
	\begin{figure}[h!]
		\centering%
		\subfloat[Hare]{{\includegraphics[width=7.5cm, height=6cm]{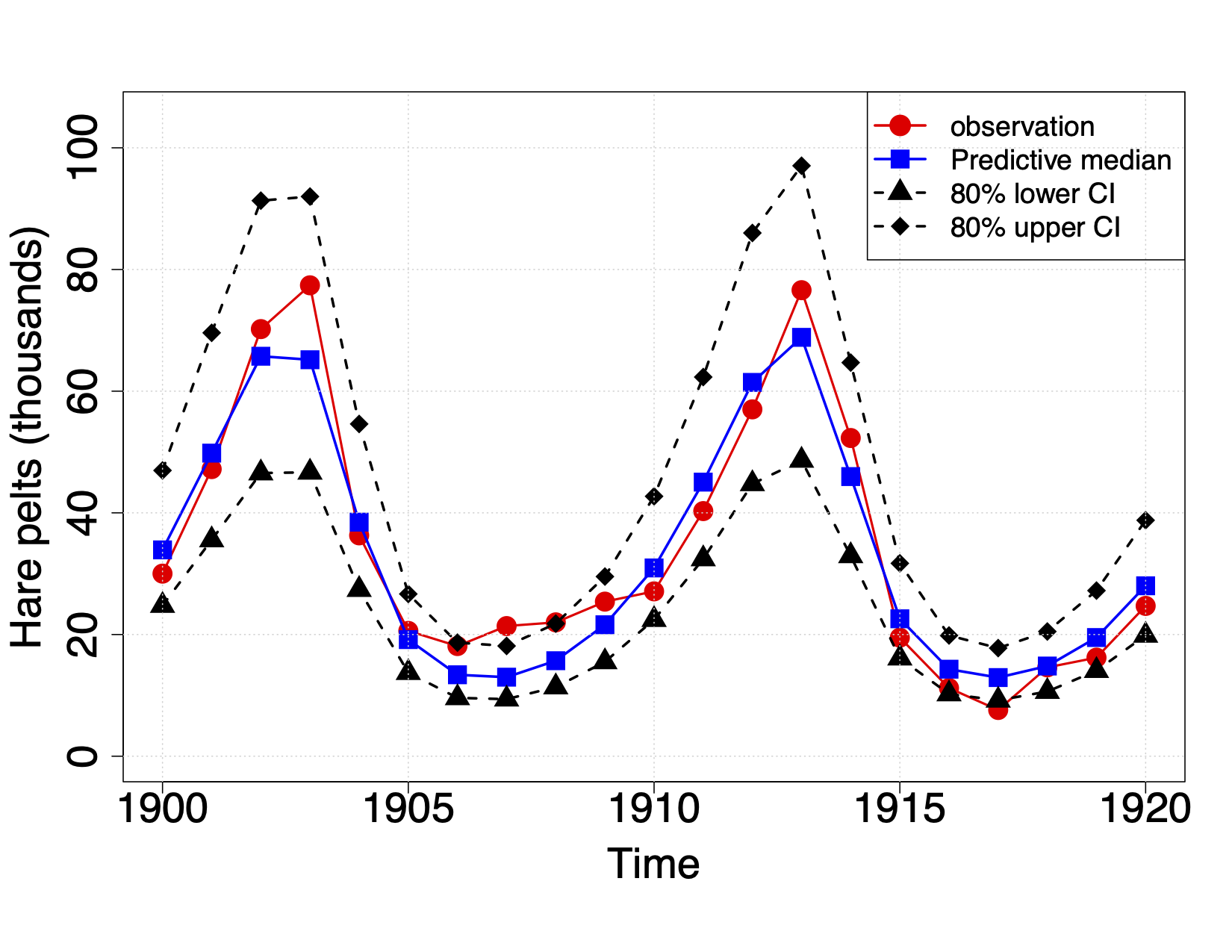}}}%
	    \qquad
		\subfloat[Lynx]{{\includegraphics[width=7.5cm, height=6cm]{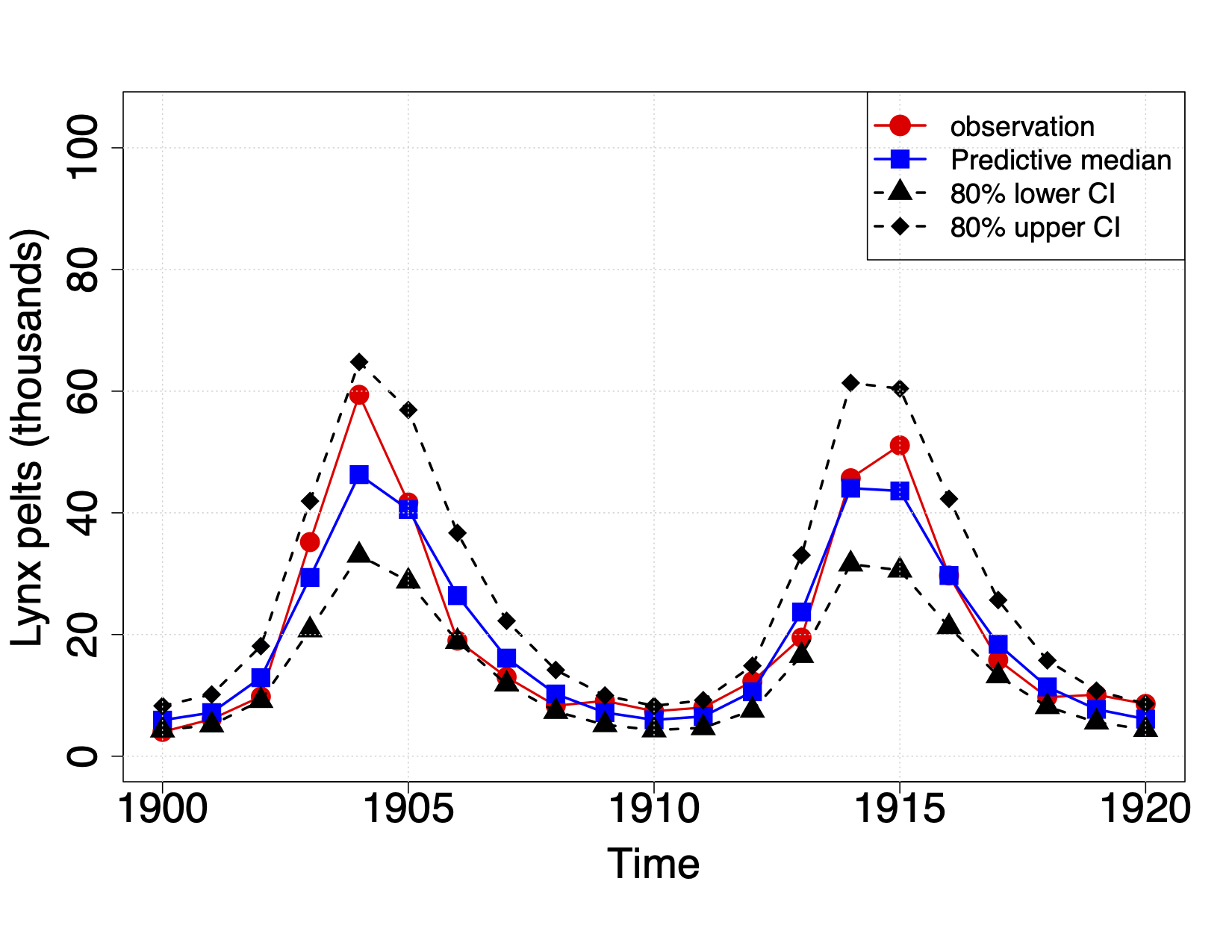} }}%
		\caption{The posterior predictive mean and $80$\% credible interval of Hare and Lynx population using the posterior samples from the TSAM sampler.}
		\label{fig:prediction}
	\end{figure}

The marginal posterior distributions of TSAM and AM are shown in Figure \ref{fig:post_pred}, which overlaps with each other confirming that both methods perform similarly in approximating the target posterior distribution. From Figure  \ref{fig:redpm_predatorPrey} (a) we observe that the autocorrelation plots of the log-posteriors are very similar for both methods. Figure \ref{fig:redpm_predatorPrey} (b) shows the REDPM of the TSAM over the AM for each of the parameters are between 5 and 7 in each of the three cases ``no thinning", ``thinning 10" and ``thinning 20''. The REDPM of TSAM over AM using the autocorrelation of the log posterior is $7.2$, much larger than 1. Thus, the TSAM algorithm is computationally very efficient when compared to the AM, while the accuracy for both methods is very similar. The posterior predictive mean and the $80$\% credible interval of hare and lynx populations are plotted in Figure \ref{fig:prediction}. The plot shows that the calibrated model using TSAM fits the observed data quite well. 	

From the four examples considered in this section, we can conclude that both the AM and TSAM methods have stronger convergence and better mixing properties than the MH and TSMH. Both the AM and the TSAM perform very similarly regarding ergodicity and convergence but the TSAM is much more computationally efficient than the AM. 

%\clearpage	

\section{Discussion and conclusion} \label{conclusion}
% Bayesian parameter calibration in a inverse problem setting has applications in many fields such as subsurface modeling, remote sensing, chemical kinetics etc. In many situations the corresponding physics based forward model is very expensive to run in a fine-grid for obtaining a desired accuracy. In such cases direct posterior sampling is almost impossible due to computational cost, and our two-stage method can help reducing the computational cost.
A new Monte Carlo sampler named two-stage adaptive Metropolis is proposed in this article, based on two very powerful ideas in MCMC literature -- Adaptive Metropolis and two-stage Metropolis-Hastings sampler. The proposed sampling algorithm is very effective in sampling from high-dimensional posterior distributions where the target density is expensive to compute and approximate inexpensive versions of it exist. The ergodicity property of the algorithm is proved theoretically. The simulation and real data examples corroborate with the theoretical result and the effectiveness of the sampler when compared to its peers. In this article, we considered only a few applications of the proposed TSAM sampler to establish its superiority over the existing methods. But the proposed sampling method will have a broader impact in many real applications, especially  in the area of Bayesian inference and calibration for large-scale computer models, where the forward model is computationally very expensive and alternative inexpensive surrogate models are easily available. Such applications very frequently arise in many fields such as subsurface modeling, remote sensing, chemical kinetics, etc. TSAM can also be easily generalized to a multi-stage adaptive Metropolis (MSAM) algorithm, where the adaptive proposals are screened at multiple ($>1$) stages before being accepted in the final stage. Such MSAM can be very useful for posterior simulations in a multi-fidelity computer model calibration setup. 

%\section{Conclusion}
%\section{Appendix}

\bibliography{references}

\end{document}